\documentclass[12pt]{scrartcl}
\usepackage{color}
\usepackage{tabularx}
\usepackage{amsmath}
\usepackage{amsthm}
\usepackage{amssymb,amsfonts,textcomp}
\usepackage{latexsym}
\usepackage{graphicx} 
\usepackage{multirow}
\usepackage{microtype}
\usepackage{rotating}
\usepackage{url}
\usepackage{color}
\usepackage{comment}
\usepackage{paralist}
\usepackage{booktabs}
\usepackage{bm}
\usepackage[]{units}
\usepackage[noend, ruled, noline]{algorithm2e}
\usepackage{tikz}
\usepackage[numbers]{natbib}

\allowdisplaybreaks

\RequirePackage{ifthen,calc}
\newsavebox{\ffbox}\newlength{\ffboxlen}
\newcommand{\todo}[1]{%
  {\sbox{\ffbox}{\textbf{TODO:}\ \textit{{#1}}\ \textbf{:ODOT}}
    \settowidth{\ffboxlen}{\usebox{\ffbox}}
		\addtolength{\ffboxlen}{-5mm}
    \ifthenelse{\ffboxlen>\linewidth}{%
      \noindent\marginpar{$>>>>$}\textbf{TODO:}\ \textit{{#1}}\ \textbf{:ODOT}\marginpar{$<<<<$}}{%
      \noindent\marginpar{$>><<$}\textbf{TODO:}\ \textit{{#1}}\ \textbf{:ODOT}}}}

\usepackage{authblk}

\usepackage{booktabs} 
\usepackage[ruled]{algorithm2e} 

\SetAlFnt{\small}
\SetAlCapFnt{\small}
\SetAlCapNameFnt{\small}
\SetAlCapHSkip{0pt}
\IncMargin{-\parindent}

\definecolor{color1}{RGB}{194, 34, 13}
\definecolor{color2}{RGB}{56, 13, 196}
\definecolor{color2}{RGB}{0,76,153}
\definecolor{color3}{RGB}{152, 194, 13}

\usepackage[colorlinks,citecolor=color2,linkcolor=color2]{hyperref}

\usepackage{xspace}
\usepackage{amsmath}
\usepackage{amsthm}
\usepackage{thmtools}
\usepackage{cleveref}
\usepackage{tikz}
\usetikzlibrary{arrows.meta}
\usetikzlibrary{arrows}
\usetikzlibrary{patterns}
\usetikzlibrary{positioning}
\usepackage{thm-restate}
\usepackage{bm}

\newtheorem{theorem}{Theorem}
\newtheorem{corollary}{Corollary}

\newtheorem{observation}{Observation}
\newtheorem*{claim*}{Claim}

\newenvironment{proofS}{%
	\proof}{\endproof}

\let\citep\cite

\newcommand{\decprob}[3]{%
  \begin{center}%
    \begin{minipage}{0.9\linewidth}%
      \textsc{#1}\\
      \textbf{Input:} #2\\
      \textbf{Output:} #3
    \end{minipage}%
  \end{center}%
}

\DeclareMathOperator*{\argmax}{argmax}

\usepackage{xcolor}
\usepackage{booktabs,colortbl, array}
\usepackage{pgfplotstable}
\pgfplotsset{compat=1.8}

\definecolor{rulecolor}{RGB}{0,76,153}
\definecolor{tableheadcolor}{gray}{0.92}

\newcommand{\topline}{ %
        \arrayrulecolor{rulecolor}\specialrule{0.1em}{\abovetopsep}{0pt}%
        \arrayrulecolor{tableheadcolor}\specialrule{\belowrulesep}{0pt}{0pt}%
        \arrayrulecolor{rulecolor}}
\newcommand{\midtopline}{ %
        \arrayrulecolor{tableheadcolor}\specialrule{\aboverulesep}{0pt}{0pt}%
        \arrayrulecolor{rulecolor}\specialrule{\lightrulewidth}{0pt}{0pt}%
        \arrayrulecolor{white}\specialrule{\belowrulesep}{0pt}{0pt}%
        \arrayrulecolor{rulecolor}}

\newcommand{\appsymb}{$\bigstar$}

\allowdisplaybreaks

\newenvironment{claimproof}{\begin{proof}}{\end{proof}}
\newcommand{\runt}{r_{\text{waw}}}
\newcommand{\waw}{\textsc{Weighted Approval Winner}\xspace}

\title{\Large Selecting Matchings via Multiwinner Voting: \\ How Structure Defeats a Large Candidate Space}

\author[1]{\large Niclas Boehmer}
\author[2]{\large Markus Brill}
\author[2]{\large Ulrike Schmidt-Kraepelin}
\affil[1]{\normalsize Research Group Algorithmics and Computational Complexity, TU Berlin}
\affil[2]{\normalsize Research Group Efficient Algorithms, TU Berlin}

\date{}

\publishers{%
    \normalfont\normalsize%
    \parbox{0.8\linewidth}{%
        Given a set of agents with approval preferences over each other, we study the task of finding $k$ matchings fairly representing {everyone's} preferences. We model the problem as an approval-based multiwinner election where the set of candidates consists of matchings of the agents, and agents' preferences over each other are lifted to preferences over matchings. 
Due to the exponential number of candidates in such elections, standard algorithms for classical sequential voting rules (such as those proposed by Thiele and Phragmén) are rendered inefficient.  
We show that the computational tractability of these rules can be regained by exploiting the structure of the approval preferences. Moreover, we establish algorithmic results and axiomatic guarantees that go beyond those obtainable in the general multiwinner setting. Assuming that approvals are symmetric, we show that \emph{Proportional Approval Voting (PAV)}, a well-established but computationally intractable voting rule, becomes polynomial-time computable, and its sequential variant (\emph{seq-PAV}), which does not provide any proportionality guarantees in general, fulfills a rather strong guarantee known as \emph{extended justified representation}. Some of our algorithmic results extend to other types of compactly representable elections with an exponential candidate space. 
    }
}

\begin{document}

\maketitle

\section{Introduction}
Matching problems involving preferences occur in a wide variety of applications, and the literature has identified a host of criteria for choosing a ``fair'' matching \cite{Manl13a}. In contrast to most of this work, we are interested in situations where \textit{multiple} matchings between agents need to be chosen based on the preferences of agents over each other.
Such situations occur naturally in applications where agents need to be matched multiple times, either successively or simultaneously. For instance,
teachers often divide students into pairs for partner work, and multiple matchings might be required for different learning activities and different subjects. Several matchings also need to be found in pair programming, if, for example, one pairing is selected per project milestone.
Other natural applications occur in workplaces where shifts are executed in pairs, which is often the case for security reasons (e.g., police officers or pilots usually work in shifts as pairs). 

We model scenarios of this type as the problem of finding $k$ matchings between agents 
based on the agents' dichotomous (i.e., approval/disapproval) preferences over each other. More concretely, we associate with each agent an approval set, i.e., a subset of other agents that are approved by the agent. 
In the student/teacher scenario, approval sets of students could, for example, consist of all students they like, or of all students that are deemed compatible by the teacher.
Preferences over agents are then lifted to preferences over matchings in a straightforward way: An agent approves a matching if and only if she is matched to an agent she approves. 
If the task were to find only a single matching, it would be natural to select a matching maximizing the number of approvers. However, as we are interested in finding multiple matchings, it is often possible to balance interests of agents by selecting a collection of matchings that makes all agents at least partially happy. Hence, our goal is to find~$k$ matchings that fairly represent the agents' preferences.

By interpreting matchings as candidates and agents as voters in an election, 
our setting can be viewed as a special case of approval-based multiwinner elections \cite{ABC+16a}. 
As a consequence, voting rules and axiomatic results from this more general framework are applicable to our setting, to which we refer to as \textit{matching elections}.
Since we aim to treat agents fairly, we are particularly interested in axioms capturing \textit{proportional representation}, and in voting rules satisfying those axioms. 
In matching elections, we explicitly allow that a single candidate (i.e., matching) can be selected multiple times. This is in contrast to general approval-based multiwinner elections, where candidates can be selected at most once. As a rationale for our decision, observe that such a constraint would be rather artificial in our setting: Two matchings which only differ in a few pairs would already be considered as two distinct candidates in a matching election. 
Allowing matchings to be selected multiple times positions matching elections within the class of \textit{party-approval elections}~\citep{BGP+19a}, a recently introduced subclass of approval-based multiwinner elections for which stronger axiomatic guarantees are obtainable.

Matching elections exhibit two characteristics that give rise to several interesting theoretical questions: %
First, the number of candidates in a matching election is exponential in the number of agents 
(and thus in the size of the description of an instance). 
As a consequence, a number of standard algorithms for applying voting rules or checking axiomatic guarantees no longer run efficiently, as they iterate over the candidate space.
Second, preferences of agents have a very specific structure. For instance, it is possible to combine certain parts of two matchings, thereby obtaining a ``compromise'' candidate that is approved by some approvers of the first and some approvers of the second matching. Exploiting this structure has the potential to not only recover the computational tractability of voting rules, but also to prove proportional representation guarantees that go beyond those obtainable in more general multiwinner settings. 

We also consider two natural special cases of matching elections: \textit{symmetric} matching elections, where agents' approvals are mutual, and \textit{bipartite} matching elections, where agents are partitioned into two groups and agents only approve members of the opposite group. The previously described applications yield symmetric matching elections if, for example, approvals encode compatibility constraints. Similarly, bipartite matching elections arise whenever matched agents are required to have different attributes regarding professional experience, educational background, gender, etc.

\subsection{Related Work}
\label{sec:related}

Recent years have witnessed a considerable amount of interest in approval-based multiwinner elections (see \cite{LaSk20a} for a survey);  a particular focus has been on axiomatic properties capturing the notion of proportional representation \cite{ABC+16a,SFF+17a,BLS16a,BFJL16a,FMS18a,CJMW19a,JMW20a,PeSk20a,AzLe21a,PPSS21a}.

A fundamental challenge in computational social choice is to model settings where agents are presented an exponential number of possibilities. One method to deal with this is to assume that there exists some compact representation of the agents' preferences that can be systematically lifted to preferences over all possibilities. This approach has been used, for example, in the study of hedonic games \cite{BoJa02a,DBLP:journals/teco/AzizBBHOP19,DBLP:conf/ijcai/BoehmerE20}, fair division \cite{BEL10a,AGM+15a}, and single-winner voting in combinatorial domains \cite{CELM08a,LaXi15a}.  To the best of our knowledge, multiwinner elections with exponentially many candidates have not yet been considered.

The literature on matching problems has produced a variety of optimality criteria for selecting a single matching based on ordinal preferences of agents \cite{Manl13a}. Interestingly, some of these criteria are based on concepts from voting theory. For instance, a popular matching \citep{Gard75a,Cseh17a} corresponds to a (weak) Condorcet winner of the corresponding single-winner matching election. Settings in which \textit{multiple} matchings are to be selected are rare in the literature; a notable exception is the literature on dynamic matching markets, which mostly focuses on extending the notion of stability to temporal settings \citep{dynamic-1,BLY20a}.

Notably, \citet{BoMo04a} consider a setting that is similar to ours, except that \textit{probability distributions} over matchings are chosen (rather than multiple matchings). Probabilities of matchings can be interpreted as time shares, and the utility of an agent is given by the probability of being matched to an approved partner.
They focus on the \textit{egalitarian solution} \citep{BoMo04a}, which chooses probability distributions maximizing the utility of the worst-off agent (breaking ties according to the leximin order).
It was recently shown that such a probability distribution can be computed in polynomial time~\citep{GaBo20a}. 
\citet{BoMo04a} only consider bipartite and symmetric\footnote{%
More precisely, \citet{BoMo04a} do allow asymmetric preferences but assume that agents can only be matched if they approve each other, effectively rendering the setting symmetric.} 
instances and show that, under these restrictions, the egalitarian solution satisfies strong fairness and incentive properties.
It is easy to verify that the egalitarian solution does \textit{not} lead to proportional outcomes in unrestricted matching elections as considered here (see \Cref{fn:egalitarian}).

\subsection{Our Contributions}
We establish matching elections as a novel subdomain of approval-based multiwinner elections with an exponential candidate space and initiate their computational and axiomatic study. 
We consider several established 
(classes of) approval-based multiwinner rules (Thiele rules, Phragmén's sequential rule, and Rule~X) and proportionality axioms 
(PJR, EJR, and core stability). 
Exploiting the structure of matching elections, we prove a number of positive results. 
In particular, we show that all considered sequential rules can be computed in polynomial time despite the exponential candidate space. 
In fact, we show the slightly more general result that those rules are tractable in all elections where a candidate maximizing a weighted approval score can be found efficiently.
We furthermore show that non-sequential Thiele rules such as PAV can be computed efficiently in symmetric and in bipartite matching elections, whereas they are computationally intractable in general matching elections (with a general matching election we mean an election that is neither bipartite nor symmetric). We present these results in \Cref{sec:cc}, which we start with a table (\Cref{fig:ov1}) summarizing our computational results. 

The additional structure of symmetric matching elections has axiomatic ramifications as well: We show that a large class of sequential Thiele rules satisfies EJR in this setting. This is particularly surprising as these rules are known to violate even significantly weaker axioms in general multiwinner elections. On the other hand, Phragmén's sequential rule and Rule~X do not satisfy stronger proportionality axioms compared to the general setting. We present these results in \Cref{sec:axioms}, which we start with a table (\Cref{fig:ov2}) summarizing our axiomatic results.

Lastly, in \Cref{sec:check-axiom},  we show that in matching elections it can be checked efficiently whether a committee satisfies EJR, whereas checking core stability or PJR is intractable. The problem of checking PJR is our only example for a computational problem that is polynomial-time solvable in the party-approval setting and NP-complete in the setting of matching elections. 

The proofs (or their completions) for results marked by (\appsymb) can be found in the appendix.

\section{Preliminaries}
In this section, we define party-approval elections and recap some approval-based multiwinner voting rules and proportionality axioms. 
Let $\mathbb{N}=\{1, 2, \ldots\}$ and $\mathbb{N}_0 = \mathbb{N} \cup \{0\}$.
For $n \in \mathbb{N}$, let $[n]$ denote the set $\{1,\ldots,n\}$.

\subsection{Party-Approval Elections}

A \emph{party-approval election} \citep{BGP+19a} is a tuple $(N,C,A,k)$, where $N$ is a set of agents, $C$ a set of candidates, $A=(A_{a})_{a\in N}$ a preference profile with $A_a\subseteq C$ denoting the approval set of agent~$a$, and $k \in \mathbb{N}$ the committee size.\footnote{To avoid trivial instances, we always assume that there exists at least one agent $a \in N$ with $A_a \neq \emptyset$.} 
A \emph{committee} $W: C \rightarrow \mathbb{N}_0$ is a multiset of candidates, with the interpretation that $W(c)$ is the number of copies of candidate~$c$ contained in $W$. The \textit{size} of a committee $W$ is given by $\sum_{c \in C} W(c)$. 
For an agent $a\in N$ and a committee $W$, we let the happiness score $h_a(W)$ of~$a$ denote the number of (copies of) candidates from $W$ approved by~$a$, i.e., $h_a(W) = \sum_{c \in A_a} W(c)$. Moreover, $N_c= \{a \in N \mid c \in A_a\}$ denotes the set of \textit{approvers} (also called supporters) of $c$, and $|N_c|$ is called the \textit{approval score} of~$c$. 
A \textit{voting rule} maps a party-approval election $(N,C,A,k)$ to a set of committees of size $k$. 
All committees output by a voting rule are considered tied for winning. 
Party-approval elections differ from the more general approval-based multiwinner elections \cite{ABC+16a} in that candidates can appear in a committee multiple times.

It is usually assumed that instances of an election are described by listing all candidates and approval sets explicitly. Since we will deal with elections with an exponential candidate space, we relax this assumption and only require that a representation of an election is given from which the full election can be reconstructed. We will show that several computational problems we consider in the following can be reduced to solving the following problem:
\decprob{Weighted Approval Winner}{A representation of a party-approval election $(N,C,A,k)$ and a weight function \mbox{$\omega\colon N\mapsto \mathbb{R}_{\geq 0}$}.}{A candidate maximizing the total weight of its approvers, i.e., an element of $\argmax_{c\in C}\sum_{a\in N_c} \omega(a)$.}
\noindent We let $\runt$ denote the running time of solving this problem.

\subsection{Voting Rules from Multiwinner Voting}

We describe several methods for computing committees. 
The output of a voting rule consists of all committees that can result from this method for some way of breaking ties.

\newcommand{\score}{\mathit{sc}}

 \paragraph{Thiele Rules \emph{\citep{Thie95a,Jans16a}}} 
 The class of $w$-Thiele rules is parameterized by a \textit{weight sequence}~$w$, i.e., an infinite sequence of non-negative numbers $w=(w_1,w_2,\dots)$ such that $w_1=1$ and $w_i\geq w_{i+1}$ for all~$i$. 
Given a weight sequence~$w$, the score of a committee $W$ is defined as $\score_w(W) = \sum_{a \in N} \sum_{i=1}^{h_a(W)}w_i$. 
The rule $w$-Thiele selects committees maximizing this score. 
Setting $w_i = 1/i$ for all $i \in \mathbb{N}$ yields the arguably most popular $w$-Thiele rule known as \textit{Proportional Approval Voting (PAV)}.   

\paragraph{Sequential $w$-Thiele Rules (seq-$w$-Thiele rules) \emph{\cite{Thie95a,Jans16a}}} 
These variants of $w$-Thiele rules start with the empty committee and add candidates iteratively. Given a multiset~$W$ of already selected candidates, the \textit{marginal contribution} of a candidate $c$ is defined as $\score_w(W \cup \{c\})-\score_w(W)$. In each step, seq-$w$-Thiele adds a candidate with a maximum marginal contribution. Setting $w_i = 1/i$ for all $i \in \mathbb{N}$, we obtain the rule \emph{seq-PAV}.

\paragraph{Phragmén's Sequential Rule (seq-Phragmén) \emph{\citep{Phra94a,Jans16a}}}
In seq-Phragmén, all agents start without money and continuously earn money (i.e., budget) 
at an equal and constant speed. As soon as there is a candidate $c$ such that the group $N_c$ jointly owns one dollar, such a candidate is added to the committee $W$ and the budget of the group $N_c$ is reduced to zero. All remaining agents keep their budget. This is repeated until the committee has size $k$.

\paragraph{Rule~X \emph{\citep{PeSk20a}}} 
Initially, every agent~$a$ has a budget~$b_a$ of $k/n$ dollars. Each candidate costs one dollar and a candidate $c$ is said to be $q$-affordable if $\sum_{a \in N_c} \min\{b_a,q\} \geq 1$. In each round, we add a candidate which is $q$-affordable for minimum $q$ and reduce the budget of the agents from $N_c$ accordingly. The rule stops when there exists no $q$-affordable candidate for any $q>0$. Note that Rule~X might create a committee of size smaller than~$k$; in this case, the committee can be completed by choosing the remaining candidates arbitrarily \citep{PeSk20a}. 

\smallskip
Since seq-$w$-Thiele rules, seq-Phragmén, and Rule~X add candidates to the committee one by one, we refer to these rules as \textit{sequential rules}.

\subsection{Axioms from Multiwinner Voting}
Consider a party-approval election $(N,C,A,k)$.
For $\ell \in [k]$, a set of agents $S \subseteq N$ is \emph{$\ell$-cohesive} if $|S|\geq \ell \frac{n}{k}$ and  $\bigcap_{a \in S} A_a \ne \emptyset$.
We consider three axioms capturing proportional representation  \cite{ABC+16a,SFF+17a}:
\paragraph{Proportional Justified Representation} A committee $W$ provides \emph{proportional justified representation (PJR)} if there does not exist $\ell\in [k]$ and an $\ell$-cohesive group $S$ such that $W$ contains strictly less than $\ell$ (copies of) candidates that are approved by at least one agent in $S$, i.e., $\sum_{c \in \bigcup_{a \in S}A_a} W(c) < \ell$.
\paragraph{Extended Justified Representation} A committee $W$ provides \emph{extended justified representation (EJR)} if there does not exist $\ell\in [k]$ and an $\ell$-cohesive group $S$ such that $h_a(W) < \ell$ for all $a \in S$. 
\paragraph{Core Stability} Given a committee $W$, we say that a group of agents $S\subseteq N$ \textit{blocks} $W$ if $|S| \geq \ell \frac{n}{k}$ for some $\ell \in [k]$ and there exists a committee $ W'$ of size $\ell$ such that $h_a(W') > h_a(W)$ for all $a \in S$. A committee $W$ is \emph{core stable} if it is not blocked by any group of agents. 
\medskip

Core stability implies EJR \citep{ABC+16a}, and EJR implies PJR \citep{SFF+17a}. 
As it is standard in the literature \citep{LaSk20a}, we say that a voting rule satisfies PJR/EJR/core stability if \textit{all} committees in its output satisfy the respective condition.

\section{Matching Elections}\label{sec:first_observation}
In this section, we formally introduce matching elections and establish them as a special case of party-approval elections by giving a formal embedding. We familiarize ourselves with the newly introduced setting by proving some first observations on the special structure of the candidate space as well as showing that the weighted approval winner problem can be solved efficiently. 

A \emph{matching election} is a tuple $(N,A,k)$, 
where $N$ is a finite set of agents, 
$A=(A_a)_{a\in N}$ a preference profile with $A_a\subseteq N\setminus\{a\}$ denoting the set of agents that are approved by agent $a$, 
and $k \in \mathbb{N}$ the number of matchings to be chosen. We let $n$ denote the number of agents $|N|$. For notational convenience, we also call $(N,A)$ a matching election.

A \emph{matching} $M$ is a subset of (unordered) pairs of agents, i.e., $M \subseteq \{\{a,b\} \mid a,b \in N, a \neq b\}$, such that no agent is included in more than one pair. If $\{a,b\} \in M$, we say that $a$ is $b$'s \emph{partner} or $a$ is \emph{matched to} $b$ in $M$. A matching~$M$ is \emph{perfect} if every agent has a partner. An agent $a$ \emph{approves} a matching $M$ if $a$ is matched to some agent~$b$ in~$M$ and $a$ approves $b$, i.e., $b\in A_a$. We let $N_M$ denote the set of agents approving matching~$M$. We call a matching~$M$ \emph{Pareto optimal} if there does not exist another matching~$M'$ such that $N_M \subsetneq N_{M'}$.
We call a matching \emph{minimal} if there does not exist another matching~$M'$ such that $M' \subset M$ and $N_M = N_{M'}$.  An outcome of a matching election is a multiset (or committee) $\mathcal{M}$ of $k$ Pareto optimal and minimal matchings. \footnote{Minimality is only a formal restriction introduced for the sake of consistency, 
as any minimal matching can be extended to a (nearly) perfect matching by adding pairs of unmatched agents.
Pareto optimality enforces that no clearly suboptimal matchings are part of the committee. We can convert any matching $M$ into a Pareto optimal matching $M'$ with $N_M \subseteq N_{M'}$ by solving one instance of \waw. For details, we refer to the proof of \Cref{le:lemWeighApp}.}

\paragraph{Approval Graph}
The \emph{approval graph} of a matching election $(N,A)$ is a mixed graph defined as follows. The nodes of the approval graph are the agents in $N$ and the edges depict the approval preferences: For two agents $a,b\in N$, there is an undirected edge $\{a,b\}$ if $a$ approves $b$ and $b$ approves $a$; 
and there is a directed edge $(a,b)$ if $a$ approves $b$ but $b$ does not approve~$a$. For an example, see the illustration on the left in \Cref{fig:ex}. Observe that a matching is minimal if and only if it contains only pairs which are connected by an (undirected or directed) edge in the approval graph.
Every minimal and Pareto optimal matching is in particular a maximal matching in the approval graph when all edges are interpreted as undirected. Observe that the reverse direction is not true, i.e., not every maximal matching in the approval graph is Pareto optimal.

\paragraph{Bipartite and Symmetric Matching Elections}
We consider two natural domain restrictions for matching elections.  
A matching election $(N,A)$ is called \emph{bipartite} if there exists a partition of the agents $N=N_1 \dot{\cup} N_2$ such that each agent approves only agents from the other set, i.e., if $a \in N_i$ for $i \in \{1,2\}$, then $A_a \subseteq N \setminus N_i$. 
Furthermore, we call a matching election $(N,A)$ \emph{symmetric} if agents' approvals are mutual, i.e., for two agents $a,b\in N$, $b\in A_a$ implies $a\in A_b$.

\subsection{Embedding into Party-Approval Elections}
\label{sec:embedding}

A matching election $(N,A,k)$ can be transformed into a party-approval election $(N',C',A',k')$ with $N'=N$ and $k'=k$, and $C'$ being the set of all Pareto optimal and minimal matchings in $(N,A)$ and $A'$ being the preference profile where each agent approves all candidates corresponding to approved matchings.
As we thereby establish matching elections as a subclass of party-approval elections, voting rules and axioms for party-approval elections directly translate to matching elections.

\tikzstyle{node}=[draw, circle, fill, inner sep = 2pt]
\tikzstyle{squared-node}=[draw, fill, inner sep = 2pt]
\tikzstyle{edge}=[->,> = latex']
\begin{figure}
    \centering
    \begin{tikzpicture}
         \node[node, label=90:$a_1$] (v1) at (0, 0) {};
         \node[node, label=270:$a_2$] (v2) at (0, -1) {};
         \node[node, label=90:$a_3$] (v3) at (1, 0) {};
         \node[node, label=270:$a_4$] (v4) at (1, -1) {};
         \node[node, label=90:$a_5$] (v5) at (2, 0) {};
         \node[node, label=270:$a_6$] (v6) at (2, -1) {};
         \draw[-{Latex[length=2mm]}, very thick] (v1)-- (v2);
         \draw[-{Latex[length=2mm]}, very thick] (v2)-- (v3);
         \draw[-{Latex[length=2mm]}, very thick] (v5)-- (v3);
         \draw[-{Latex[length=2mm]}, very thick] (v6)-- (v4);
         \draw (v3) edge[very thick] (v4);

    \end{tikzpicture} \hspace{0.6cm}\vline\hspace{0.6cm}
        \begin{tikzpicture}
         \node[node, label=90:$a_1$] (v1) at (0, 0) {};
         \node[node, label=270:$a_2$] (v2) at (0, -1) {};
         \node[node, label=90:$a_3$] (v3) at (1, 0) {};
         \node[node, label=270:$a_4$] (v4) at (1, -1) {};
         \node[node, label=90:$a_5$] (v5) at (2, 0) {};
         \node[node, label=270:$a_6$] (v6) at (2, -1) {};
         \draw[-{Latex[length=2mm]}, very thick,color1] (v1) to[bend right=15] (v2);
        \draw[-{Latex[length=2mm]}, very thick,color2] (v1) to[bend left=15] (v2);
         \draw[-{Latex[length=2mm]}, very thick,color3] (v2)-- (v3);
         \draw[-{Latex[length=2mm]}, very thick,color2] (v5)-- (v3);
         \draw[-{Latex[length=2mm]}, very thick,color2] (v6) to[bend left=15] (v4);
        \draw[-{Latex[length=2mm]}, very thick,color3] (v6) to[bend right=15] (v4);
         \draw (v3) edge[very thick,color1] (v4);
         \node[label=\textcolor{color1}{$\bm{c_1}$}] (v1) at (3, -0.5) {};
         \node[label=\textcolor{color2}{$\bm{c_2}$}] (v2) at (3, -1) {};
         \node[label=\textcolor{color3}{$\bm{c_3}$}] (v3) at (3, -1.5) {};
    \end{tikzpicture}
    \caption{The figure on the left depicts the approval graph of the matching election $(N,A)$ with $N=\{a_1,\dots, a_6\}$ and approval sets $A_{a_1}=\{a_2\}$, $A_{a_2}=\{a_3\}$,
    $A_{a_3}=\{a_4\}$,
    $A_{a_4}=\{a_3\}$,
    $A_{a_5}=\{a_3\}$, and
    $A_{a_6}=\{a_4\}$. 
   The figure on the right depicts the three candidates $c_1$, $c_2$, and $c_3$ in the corresponding party-approval election.}
    \label{fig:ex}
\end{figure}
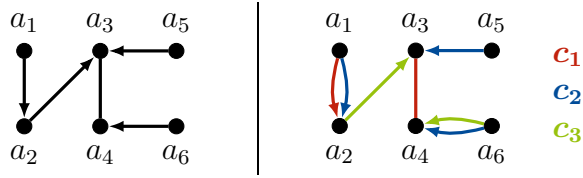

To illustrate the described transformation, we convert the matching election with six agents, whose approval graph is depicted in \Cref{fig:ex}, into a party-approval election.  
The candidates of the corresponding party-approval election are the three Pareto optimal and minimal matchings $c_1=\{\{a_1,a_2\},\{a_3,a_4\}\}$, $c_2=\{\{a_1,a_2\},\{a_3,a_5\},\{a_4,a_6\}\}$, and $c_3=\{\{a_2,a_3\},\{a_4,a_6\}\}$, which are marked on the right side of \Cref{fig:ex}. The approval sets of the agents in the corresponding party-approval election are $A_{a_1}=\{c_1,c_2\}$, $A_{a_2}=\{c_3\}$, $A_{a_3}=A_{a_4}=\{c_1\}$, $A_{a_5}=\{c_2\}$, and $A_{a_6}=\{c_2,c_3\}$. 

To get a feeling for proportionality in this election, let us set $k=3$. Observe that the groups $\{a_3,a_4\}$ and $\{a_5,a_6\}$ make up one third of the electorate while at the same time, each of the groups can agree on a matching they commonly approve. In other words, both groups are $1$-cohesive. Since $a_3$ and $a_4$ only approve $c_1$, this is a strong argument in favor of choosing $c_1$ at least once. Given that $c_1$ is chosen at least once, adding $c_2$ seems preferable over adding $c_3$, since $c_2$ is approved by three agents, two of which are completely unhappy so far, whereas $c_3$ is approved by only two so far completely unhappy agents.
Lastly, there is the choice between selecting $c_3$, which would lead to every agent being satisfied at least once, and selecting one of the more popular matchings $c_1$ or $c_2$ again. In fact, all three resulting committees are core stable.  
PAV and seq-PAV both select $\{c_1,c_2,c_3\}$ in this example, whereas seq-Phragmén returns  $\{c_1,c_2,c_3\}$ and $\{c_1,c_1,c_2\}$ as tied winners. Rule~X terminates after adding $c_1$ and $c_2$ to the committee, which can be interpreted as a three-way tie between $\{c_1,c_1,c_2\}$, $\{c_1,c_2,c_2\}$, and $\{c_1,c_2,c_3\}$.%
\footnote{\label{fn:egalitarian}
A modified version of this example can be used to show that the egalitarian solution \citep{BoMo04a} (see \Cref{sec:related}) may fail to select proportional outcomes:  
Consider the matching election that results from the one depicted in \Cref{fig:ex} when restricting the set of agents to $\{a_1,a_2,a_3,a_4\}$. This election has two candidates $c =\{\{a_1,a_2\},\{a_3,a_4\}\}$ and $c' = \{\{a_2,a_3\}\}$ and approval sets  
$A_{a_1}=A_{a_3}=A_{a_4}=\{c\}$ and $A_{a_2}=\{c'\}$.
The egalitarian solution selects the uniform probability distribution over~$\{c,c'\}$, which can be interpreted in our setting as selecting equally many copies of $c$ and $c'$ (for even $k$). 
This, however, violates PJR. To see this, let $k=4$ and consider the group $S=\{a_1,a_3,a_4\}$. This group is $3$-cohesive, but the committee $\{c,c,c',c'\}$ corresponding to the egalitarian solution contains only $2$ copies of the only candidate approved by agents in~$S$.}

While the focus of this paper is on matching elections, we note that some of our results apply to general party-approval elections. In particular, we establish our algorithmic results in \Cref{sec:comp-rules,sec:check-axiom} by reducing the computational problem at hand to solving instances of \waw (which is polynomial-time solvable for matching elections as shown in \Cref{sub:waw-match}). 

\subsection{First Observations on the Candidate Space}\label{sub:firstOb} 

In this subsection, we make some general first observations about features of our candidate space and the agents' approval sets. We start with an observation about the richness of the candidate space. Given a candidate (i.e., a matching) $M$ and an agent~$a$ disapproving~$M$, it is possible to obtain a new candidate $M'$ that is approved by $a$ and by all agents approving $M$ except at most three: Assuming that $a$ approves at least one agent, say~$b$, to construct $M'$, we remove the pair from $M$ containing $b$, say $\{b,c\}$ (if it exists), as well as the pair containing $a$, say $\{a,d\}$ (if it exists). Finally, we insert the pair $\{a,b\}$. Observe that, for the approval of $a$, we lost at most three approvals from~$M$, namely the ones of $b,c$, and $d$.

\begin{observation}
Given a matching election $(N,A)$ with ${n \geq 4}$, let $M$ be a matching and 
$a \in N \setminus N_M$ an agent with $A_a \neq \emptyset$. 
There exists a matching $M'$ which is approved by $a$ and all but at most three agents from $N_M$.
\end{observation}

Using this exchange argument, it is possible to show that the number of approvals of each Pareto optimal matching $M$ is at least $\frac{1}{3}$ of the number of approvals of any other matching $M'$. To see this, we create a third matching $\widetilde{M}$ by the following procedure: Initially, set $\widetilde{M}=M'$. As long as there exists an agent $a \in N_{M}$ not approving $\widetilde{M}$, insert into $\widetilde{M}$ the pair from $M$ containing $a$, say $\{a,b\}$, and delete the pairs $\{a,d\},\{b,c\}$ from~$\widetilde{M}$ (if they exist). This procedure terminates in $|N_M|$ steps, since every agent in $N_M$ is considered at most once. After termination, $N_M \subseteq N_{\widetilde{M}}$ and because $M$ is Pareto optimal, also $N_{M}=N_{\widetilde{M}}$. Since in each iteration the number of approvals went down by at most two, we get that $|N_{M}| \geq |N_{M'}| - 2|N_{M}|$. 
\begin{observation}\label{ob:onethird}
    Let $(N,A)$ be a matching election and $M$ a Pareto optimal matching. For any other matching $M'$, it holds that $|N_M| \geq \frac{1}{3} |N_{M'}|$. 
\end{observation}
Thus, we know that all candidates in a matching election are approved by the same number of agents up to a factor of three. For symmetric matching elections, it is possible to tighten this bound: Here, all candidates are approved by the same number of agents and it is possible to perform one-to-one exchanges. 
This is also the key observation that helps proving that many seq-$w$-Thiele rules satisfy EJR.
To see why this is true, recall that, in a symmetric matching election, the set of agents approving a minimal matching is exactly the set of matched agents. For the sake of contradiction, assume that there exist two minimal Pareto optimal matchings $M$ and $M'$ where $M$ matches more agents than $M'$. Then, the symmetric difference of $M$ and $M'$ contains at least one path of odd length starting and ending with an edge from $M$. By augmenting $M'$ along this path, it is possible to match an additional agent, which contradicts that $M'$ is Pareto optimal. 
\begin{observation} \label{ob:equalapproval}
    In symmetric matching elections, all candidates have the same approval score and correspond to maximum matchings in the approval graph.
\end{observation}

The first part of \Cref{ob:equalapproval} already implies that symmetric matching elections have a strong structure. The second part has even further implications on the distribution of approvals of agents. These follow from the Gallai-Edmonds Structure Theorem \citep{Gall64a,Edmo65a}, which describes the structure of maximum matchings in undirected graphs. For our setting, the theorem implies that we can partition the agents into three sets $W,X,$ and $Y$ such that all agents from $X$ and $W$ approve every Pareto optimal matching. Moreover, in every Pareto optimal matching, all agents from $X$ are matched to agents from $Y$ and agents from $W$ are matched among themselves. Using this theorem, we can convert every symmetric matching election into an essentially equivalent bipartite matching election. Here, the agents $Y$ form one part of the bipartition and agents from $X$ (plus some dummy nodes) form the other part. We present details on the transformation and the Gallai-Edmonds Structure Theorem in \Cref{sub:biptosym} and in \Cref{ap:cc}.

\medskip

\subsection{Weighted Approval Winner Problem}\label{sub:waw-match} 

For matching elections, we can solve \waw efficiently by solving two \textsc{Maximum Weighted Matching} instances:

\begin{restatable}[\appsymb]{lemma}{lemWeighApp}\label{le:lemWeighApp}
    Given a matching election $(N,A)$ and a weight function $\omega$, \textsc{Weighted Approval Winner} is solvable in $\mathcal{O}(n^3)$-time.
\end{restatable}

\begin{proofS}
    Given a matching election $(N,A)$ and a weight function $\omega$ on the agents, let $G$ be the undirected graph induced by the approval graph of $(N,A)$ (i.e., there exists an edge between two nodes in $G$ if there exists a corresponding directed or undirected edge in the approval graph). We define a weight function $w$ on the edges of $G$ such that for every matching $M$ in $G$ it holds that $\sum_{e \in M} w(e) = \sum_{ a \in N_M} \omega(a)$. This is achieved by summing up the weights of the endpoints approving an edge. Clearly, if $M$ is a maximum weight matching with respect to $w$, then it also maximizes $\sum_{a \in N_M} \omega(a)$, which we refer to as the \emph{weighted approval score}.
    By construction, $M$ is minimal. However, it is not guaranteed to be Pareto optimal, as there might exist agents $a \in N$ with $\omega(a) =0$. Therefore, $M$ might not be a candidate in the matching election.
    In a second step, we introduce a second weight function $\omega'$ on the agents giving all agents in $N_M$ a weight of $n+1$, and all agents in $N \setminus N_M$ a weight of $1$. Again, we derive a weight function on the edges of $G$, $w'$, guaranteeing $\sum_{e \in M} w'(e) = \sum_{ a \in N_M} \omega'(a)$. We show: If $M'$ is a maximum weight matching with respect to $w'$, then $M'$ is Pareto optimal and minimal. Moreover, $N_M \subseteq N_{M'}$ holds by construction of $\omega'$. Hence, $M'$ also maximizes the weighted approval score with respect to $\omega$. Thus, $M'$ is a solution to the \textsc{Weighted Approval Winner} problem for the matching election $(N,A)$ and the weight function $\omega$. 
\end{proofS}

\newcommand{\A}{\hat{E}}

Note that there exist other elections with an exponential candidate space for which \textsc{Weighted Approval Winner} is polynomial-time solvable. For instance, for all party-approval elections $(N,C,A,k)$ where the independent set system $(N,\{S\mid S\subseteq N_c \textnormal{ for some } c\in C\})$ forms a matroid,
\textsc{Weighted Approval Winner} reduces to finding a maximum weight independent set. This problem is polynomial-time solvable if the independence of a set $S\subseteq N$ can be checked efficiently \cite{10.5555/2190621}.

\section{Computational Complexity of Winner Determination}\label{sec:cc}
In this section, we analyze the computational complexity of computing winning committees for different voting rules. We give an overview of our results from this section in \Cref{fig:ov1}. While some of our results are tailored to matching elections, our algorithmic results in \Cref{sec:comp-rules} are applicable to a wider class of elections with an exponential number of candidates.
We start by considering sequential rules before we turn to $w$-Thiele rules. For $w$-Thiele rules, we first consider the general then the bipartite and lastly the symmetric setting. 

\begin{table}
    \centering
    \scalebox{0.85}{
\begin{tabular}{llll}
\topline
\rowcolor{tableheadcolor}
\textbf{Rules} & Party-Approval Elections & Matching Elections & Symmetric Matching Elections \\
\midtopline
$w$-Thiele & NP-hard \cite{BGP+19a}& NP-hard (Theorem \ref{th:PAV-hard}) & P (Theorem \ref{thm:wThiele-bipartite}, Corollary \ref{co:wThieleSym}) \\
seq-$w$-Thiele & P \cite{AGG+15a} 
& P (Observation \ref{obs:seqThielePoly}) & P \\
seq-Phragmén & P \cite{BFJL16a}  & P (Theorem \ref{thm:SeqPhragmen})& P \\
Rule~X & P \cite{PeSk20a}  & P (Theorem \ref{thm:RuleX}) & P \\
\bottomrule 
\end{tabular}
}
 \caption{Summary of results on the complexity of computing a winning committee for several multiwinner voting rules. We remark that the previously known results within the setting of party-approval elections do not have any implications for the matching election setting. Our hardness result (Theorem~\ref{th:PAV-hard}) is restricted to $w$-Thiele rules satisfying $w_1 >w_2>0$. We additionally prove in \Cref{thm:wThiele-bipartite}  that a winning committee under a $w$-Thiele rule in a bipartite matching election can be computed in polynomial time.} 
    \label{fig:ov1}
\end{table}

\subsection{Sequential Rules}
\label{sec:comp-rules}

For all considered sequential voting rules, we show that finding the next candidate to be added to the committee reduces to solving \waw. 
Recall that $\runt$ denotes the running time of solving the latter problem.

For sequential $w$-Thiele rules, 
this reduction is straightforward:
Given a multiset $W$ of already selected candidates, we set the weight of an agent $a$ to its marginal contribution to the score in case that a candidate in $A_a$ is added to $W$, i.e., $\omega(a)=w_{h_a(W)+1}$. The candidate returned by \waw is then added to the committee.

\begin{observation}\label{obs:seqThielePoly}
Given a party-approval election $(N,C,A,k)$ and a weight sequence $w$, a committee that is winning under seq-$w$-Thiele can be computed in $\mathcal{O}(k\cdot \runt)$-time.
\end{observation}

We show in the appendix that a similar reduction also works for a local search variant of PAV~\citep{AEH+18a}. As this variant satisfies core stability in party-approval elections \cite{BGP+19a}, a core-stable outcome in a matching election can be computed efficiently.
\begin{restatable}[\appsymb]{observation}{lpav} 
    Given a party-approval election $(N,C,A,k)$, a committee satisfying core stability can be computed in $\mathcal{O}(nk^4\ln(k)\cdot r_{waw})$-time.
\end{restatable}

Our algorithm for Phragmén's sequential rule employs \waw in a more involved way. 
\begin{theorem}\label{thm:SeqPhragmen}
Given a party-approval election $(N,C,A,k)$, a committee that is winning under seq-Phragmén can be computed in $\mathcal{O}(kn\cdot r_{waw})$-time.
\end{theorem}
\begin{proof}
In each iteration, the problem of finding a candidate to be added to the committee can be described as follows. Each agent $a \in N$ has previously accumulated a budget of $\beta_a\geq 0$ and constantly earns additional money. Thus, at time $t\in \mathbb{R}_{\geq 0}$, agent~$a$ owns $b_a(t) =\beta_a + t$ dollars. The total budget of the approvers $N_c$ of a candidate $c\in C$ at time~$t$ can be expressed as an affine linear function $f_c(t)=\sum_{a \in N_c}\beta_a + |N_c| \cdot t$. Moreover, we define $f(t) = \max_{c\in C}f_c(t)$ as the \emph{optimal value curve}, taking the value of the maximum budget of any supporter group for a candidate at time $t$. Define $t^*$ as the minimum value $t \in \mathbb{R}_{\geq 0}$ such that $f(t) = 1$. Such a value always exists and lies in the real interval $[0,1]$ since $f(0)\leq 1$ (by definition of seq-Phragmén), $f(1)\geq 1$, and $f(t)$ is continuous on $[0,1]$.
A candidate $c^*$ with $f_{c^*}(t^*) = f(t^*)=1$ is a feasible choice under seq-Phragmén in this iteration. See \Cref{fi:SeqPhragmen} for an illustration. In the following we argue that $t^*$ and $c^*$ can be computed by using a classical method from parametric optimization and solving \waw as a subroutine.

\begin{figure}\centering
\scalebox{1}{
\begin{tikzpicture}
\node(y) at (0,3){};
\node(x) at (7,0.2){};
\node at (7.2,-0.1){$t$};
\draw[->,black, ultra thick] (0,0.2) -- (y);
\draw[->,black, ultra thick] (0,0.2) -- (x);

\node at (5,3){\textcolor{color2}{$\bm{f(t)}$}};

\draw[-,black!70, thick] (0,1.2) -- (7,1.4);
\draw[-,color2,ultra thick] (0,1.2) -- (1.4,1.25);

\draw[-,black!70, thick,dash pattern = on 5pt off 2pt] (0,.9) -- (1.4,1.25);
\draw[-,black!70, thick,dash pattern = on 5pt off 2pt] (3.77,1.81) -- (7,2.6);
\draw[-,color2,ultra thick,dash pattern = on 5pt off 2pt] (1.4,1.25) -- (3.77,1.81);

\draw[-,black!70, thick, dotted] (0,0.2) -- (3.77,1.81);
\draw[-,color2,ultra thick,dotted] (3.77,1.81) -- (7,3.2);

\draw[-,color1,very thick] (0,1.52) -- (7,1.52);
\draw[-,color1,very thick] (2.6,1.52) -- (2.6,0.2);
\node at (-.3,1.6){\textcolor{color1}{$1$}};
\node at (2.6,-.1){\textcolor{color1}{$\bm{t^*}$}};

\node[fill=black!70,circle,inner sep=2](first)at(1.4,1.25){};
\node[fill=black!70,circle,inner sep=2](first)at(3.77,1.81){};

\draw [draw=black!40, thick] (8,1.8) rectangle (10,4);
\draw[-,color2,ultra thick] (8.2,3.5) -- (8.6,3.5);
\draw[-,black!50,very thick] (8.2,3.4) -- (8.6,3.4);
\node at (9.3,3.4){$f_{c_1}(t)$};
\draw[-,color2,ultra thick,dotted] (8.2,2.9) -- (8.6,2.9);
\draw[-,black!50,very thick,dotted] (8.2,2.8) -- (8.6,2.8);
\node at (9.3,2.8){$f_{c_2}(t)$};
\draw[-,color2,ultra thick,dash pattern=on 5pt off 2pt] (8.2,2.3) -- (8.6,2.3);
\draw[-,black!50,very thick,dash pattern=on 5pt off 2pt] (8.2,2.2) -- (8.6,2.2);
\node at (9.3,2.2){$f_{c_3}(t)$};

\end{tikzpicture}
}
\caption{Illustration of the situation in the proof of \Cref{thm:SeqPhragmen}. The example depicts the budget curves for three different candidates $c_1,$ $c_2,$ and $c_3$. The functions $f_{c_1}(t),$ $f_{c_2}(t)$, and $f_{c_3}(t)$ are depicted by a solid, dotted, and dashed line, respectively. The optimal value curve $f(t)$ is marked in blue.} \label{fi:SeqPhragmen}
\end{figure}
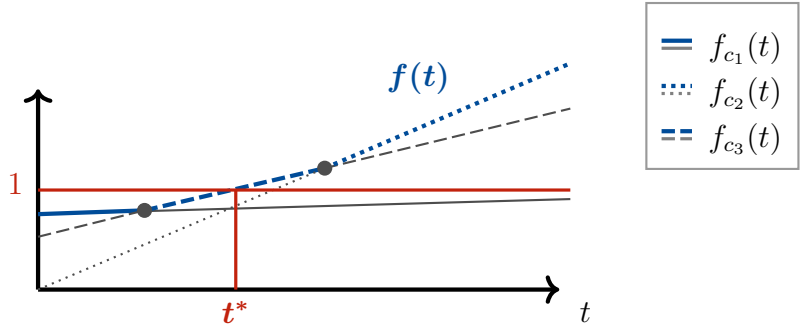

Observe that the function $f(t)$ is increasing, piecewise linear, and convex, where the latter holds because taking the pointwise maximum of convex functions results in a convex function. In order to make parts of this proof also applicable to the proof of \Cref{thm:RuleX}, we are only going to use that $f(t)$ is non-decreasing (and not that it is increasing) in the following. For a given point $t\geq 0$, we can evaluate $f(t)$ by employing the \waw problem using $b_a(t)$ as the weight of each agent $a\in N$ and computing the weight of the returned candidate. This also yields a candidate~$c$ with $f_{c}(t) = f(t)$.

The crux of finding $t^*$ is that $f(t)$ is the maximum of exponentially many functions. However, we observe that the piecewise linear function $f(t)$ has at most $n$ breaking points because the slope of $f(t)$ can take at most $n+1$ different values: for each candidate~$c$, $|N_c| \in \{0,\dots,n\}$. Hence, if we knew all breaking points of $f(t)$, we could find $t^*$ by evaluating the resulting $\mathcal{O}(n)$ linear subintervals of $f(t)$.  The Eisner-Severance method \citep{EiSe76a} can be employed to find the breaking points of $f(t)$, using $\mathcal{O}(n)$ calls to $\waw$.\footnote{More formally, the Eisner-Severance method takes as input a piecewise linear convex function $g$ defined on an interval $\mathcal{I}$. Additionally, we need to include a method to \emph{evaluate} $g$ on some point $\tau \in \mathcal{I}$ which returns the value $g(\tau)$ and an affine linear function $h$ with $g(\tau) = h(\tau)$ and $g(\tau') \geq h(\tau')$ for all $\tau' \in \mathcal{I}$. Given this input, the Eisner Severance method finds all breaking points of $g$ on $\mathcal{I}$ using $\mathcal{O}(z)$ evaluations, where $z$ is the number of breaking points of $g$.}

Since $f(t)$ is non-decreasing, we do not always have to find all breaking points in order to find $t^*$. Even though this does not improve the worst-case running time, we describe below an algorithm to find $t^*$, which mixes the idea of the Eisner-Severance method with a binary search approach.

We start by searching for two candidates $\underline{c}$ and $\overline{c}$ 
such that $f_{\underline{c}}(0)=f(0)$ and $f_{\overline{c}}(1)=f(1)$ (by solving the \waw problem). 
If $f_{\underline{c}}(0)=1$, we are done. Moreover, if $f_{\underline{c}}(t) = f_{\overline{c}}(t)$ for all $t \in [0,1]$, then there is no breaking point of $f(t)$ within the interval $[0,1]$ and we can find~$t^*$ by solving $f_{\underline{c}}(t^*) = 1$. 
Otherwise, we calculate the intersection point of $f_{\underline{c}}(t)$ and $f_{\overline{c}}(t)$, say $\hat{t}$. 
By definition of $f$, we have $f_{\underline{c}}(\hat{t}) \leq f(\hat{t})$ and we distinguish the following two cases:

If $f_{\underline{c}}(\hat{t})=f(\hat{t})$, we have found a breaking point of $f(t)$ and there is no other breaking point within the intervals $[0,\hat{t}]$ or $[\hat{t},1]$. 
Then, if $f(\hat{t}) \geq 1$, we find $t^*$ by solving $f_{\underline{c}}(t^*)=1$, and if $f(\hat{t}) < 1$, we find $t^*$ by solving $f_{\overline{c}}(t^*)=1$. 

If $f_{\underline{c}}(\hat{t})<f(\hat{t})$, we find a candidate $\hat{c}$ such that $f_{\hat{c}}(\hat{t})=f(\hat{t})$ (by solving \waw). Then, if $f(\hat{t})\geq 1$, we repeat the process for the pair $\{\underline{c},\hat{c}\}$ and the interval $[0,\bar{t}]$. If $f(\hat{t})<1$, we repeat the process for the pair $\{\hat{c},\overline{c}\}$ and the interval $[\bar{t},1]$.
We can restrict ourselves to searching within one of the two intervals because $f(t)$ is non-decreasing. This recursive procedure yields a worst-case running time of $\mathcal{O}(n \cdot r_{waw})$, as we might iterate over all breaking points. 

We have to execute the above procedure for each candidate to be added to the committee, and thus $k$ times in total. This leads to an overall running time of $\mathcal{O}(kn~\cdot~r_{waw})$. \end{proof}

In the previous proof, we upper bounded the number of breaking points of the optimal value curve by $n$. In the case of matching elections, this bound can be strengthened to $\lfloor (2/3) n \rfloor$ by making use of \Cref{ob:onethird}: The difference of the approval scores of two Pareto optimal matchings is at most $\lfloor(2/3)n \rfloor$, hence $f(t)$ takes at most $\lfloor (2/3)n\rfloor +1$ distinct slopes.

By slightly modifying the above approach, we obtain a similar algorithm for Rule~X. Here, for some fixed budgets of the agents, we need to find the minimum $q\in \mathbb{R}$ such that the supporters of some candidate jointly have one dollar, assuming that each of them pays at most $q$. We again define the optimal value curve as the maximum budget of all supporter groups dependent on $q$. Unfortunately, in this case, the optimal value curve may neither be concave nor convex. However, by observing that we can partition the domain into $n$ intervals such that the optimal value curve is a convex function in each interval, we can solve the problem using the Eisner-Severance method as described in the previous proof.

\begin{theorem}\label{thm:RuleX}
Given a party-approval election $(N,C,A,k)$, a committee that is winning under Rule~X can be computed in $\mathcal{O}(kn\cdot r_{waw})$-time.
\end{theorem}

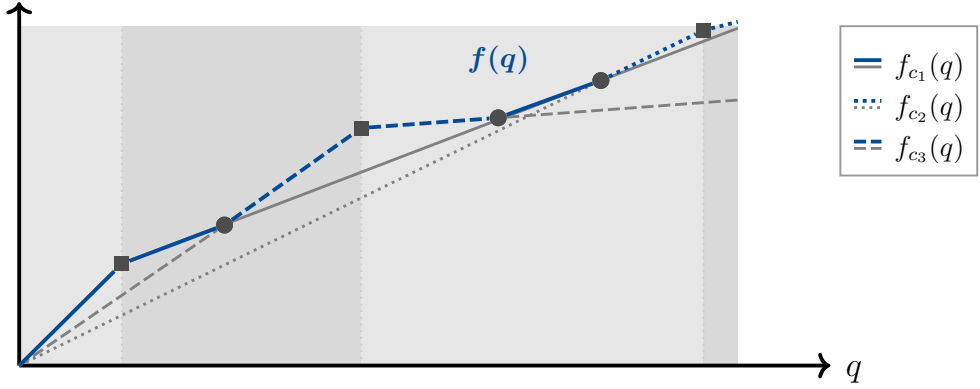
\begin{figure} \centering
\scalebox{.9}{
\begin{tikzpicture}
\node(y) at (0,5.5){};
\node(x) at (12,0){};
\node at (12.2,-0.1){\large$q$};

\draw [draw=none,fill=black!10] (0,0) rectangle (1.5,5);
\draw [draw=none,fill=black!15] (1.5,0) rectangle (5,5);
\draw [draw=none,fill=black!10] (5,0) rectangle (10,5);
\draw [draw=none,fill=black!15] (10,0) rectangle (10.5,5);

\draw[-,black!20, thick,dotted] (1.5,0) -- (1.5,5);
\draw[-,black!20, thick,dotted] (5,0) -- (5,5);
\draw[-,black!20,thick,dotted] (10,0) -- (10,5);

\draw[->,black, ultra thick] (0,0) -- (y);
\draw[->,black, ultra thick] (0,0) -- (x);

\node[fill=black!70,rectangle,inner sep=3](1)at(1.5,1.5){};
\node[fill=black!70,circle,inner sep=2.5](2)at(3,2.07){};
\node[fill=black!70,rectangle,inner sep=3](3)at(5,3.5){};
\node[fill=black!70,circle,inner sep=2.5](4)at(7,3.65){};
\node[fill=black!70,circle,inner sep=2.5](5)at(8.5,4.2){};
\node[fill=black!70,rectangle,inner sep=3](6)at(10,4.94){};

\draw[-,black!50, very thick] (0,0) -- (1);
\draw[-,black!50, very thick] (1) -- (5);
\draw[-,black!50, very thick] (5) -- (10.5,4.97);
\draw[-,black!50, very thick,dash pattern=on 5pt off 2pt] (0,0) -- (2);
\draw[-,black!50, very thick,dash pattern=on 5pt off 2pt] (4) -- (10.5,3.912);
\draw[-,black!50, very thick,dotted] (0,0) -- (5);

\draw[-,color2,ultra thick] (0,0) -- (1);
\draw[-,color2,ultra thick] (1) -- (2);
\draw[-,color2,ultra thick,dash pattern=on 5pt off 2pt] (2) -- (3);
\draw[-,color2,ultra thick,dash pattern=on 5pt off 2pt] (3) -- (4);
\draw[-,color2,ultra thick] (4) -- (5);
\draw[-,color2,ultra thick,dotted] (5) -- (6);
\draw[-,color2,ultra thick,dotted] (6) -- (10.5,5.064);

\node[fill=black!70,rectangle,inner sep=3](1)at(1.5,1.5){};
\node[fill=black!70,circle,inner sep=2.5](2)at(3,2.07){};
\node[fill=black!70,rectangle,inner sep=3](3)at(5,3.5){};
\node[fill=black!70,circle,inner sep=2.5](4)at(7,3.65){};
\node[fill=black!70,circle,inner sep=2.5](5)at(8.5,4.2){};
\node[fill=black!70,rectangle,inner sep=3](6)at(10,4.94){};

\draw [draw=black!40, thick] (12,2.8) rectangle (14,5);
\draw[-,color2,ultra thick] (12.2,4.5) -- (12.6,4.5);
\draw[-,black!50,very thick] (12.2,4.4) -- (12.6,4.4);
\node at (13.3,4.4){$f_{c_1}(q)$};
\draw[-,color2,ultra thick,dotted] (12.2,3.9) -- (12.6,3.9);
\draw[-,black!50,very thick,dotted] (12.2,3.8) -- (12.6,3.8);
\node at (13.3,3.8){$f_{c_2}(q)$};
\draw[-,color2,ultra thick,dash pattern=on 5pt off 2pt] (12.2,3.3) -- (12.6,3.3);
\draw[-,black!50,very thick,dash pattern=on 5pt off 2pt] (12.2,3.2) -- (12.6,3.2);
\node at (13.3,3.2){$f_{c_3}(q)$};

\node at (7,4.5){\textcolor{color2}{$\bm{f(q)}$}};
\end{tikzpicture}
} 
\caption{Illustration of the situation in the proof of \Cref{thm:RuleX}. The example depicts budget curves for candidates $c_1,$ $c_2,$ and $c_3$. The functions $f_{c_1}(q),$ $f_{c_2}(q)$, and $f_{c_3}(q)$ are depicted by solid, dotted, and dashed lines, respectively. The optimal value curve $f(q)$ is marked in blue. Breaking points of type $(i)$ and $(ii)$ are marked by squares and circles, respectively. Intervals in which the optimal value curve is convex are marked by gray rectangles.} \label{fig:RuleX} 
\end{figure}
\begin{proof}
At any of the iterations within the execution of Rule~X, the problem of finding a next candidate $c^*$ to be added to the committee (or deciding to stop) can be described as follows: Each agent $a\in N$ has some leftover budget $b_a \leq k/n$ at the beginning of the iteration. Then, the budget of the supporters of a candidate $c \in C$ under the restriction that every agent pays at most $q \in \mathbb{R}$ can be expressed as $f_c(q)= \sum_{a \in N_c} \min \{b_a,q\}$. Similarly as in the proof of \Cref{thm:SeqPhragmen}, we define the \emph{optimal value curve} as $f(q)=\max_{c \in C} f_c(q)$. If $f(k/n)<1$, there exists no $q$-affordable candidate for any $q$ and Rule~X terminates.  Otherwise, we aim to find the minimum $q^*$ in the real interval $[0,k/n]$ such that $f(q^*) = 1$. Such a value exists because $f(0)= 0$, $f(k/n)\geq 1$ (by the above assumption), and $f(q)$ is continuous on $[0,k/n]$. Then, a candidate $c^*$ satisfying $f_{c^*}(q^*)=1$ is a feasible next choice for Rule~X. Given $q^*$, such a $c^*$ can be found by one call to \textsc{Weighted Approval Winner}.

Observe that $f(q)$ is non-decreasing, since $f_c(q)$ is non-decreasing for all $c \in C$. However, in contrast to the proof of \Cref{thm:SeqPhragmen}, $f(q)$ is in general neither convex nor concave. 
As a consequence, we cannot directly apply the Eisner-Severance method. See \Cref{fig:RuleX} for an illustration. 
More concretely, consider some $q' \in [0,k/n]$ at which $f(q)$ has a breaking point. Then, this breaking point is of one of two types: either $(i)$ there exists some agent $a \in N$ with $b_a = q'$, or $(ii)$ such an agent does not exist. Intuitively, breaking points of type $(i)$ can (but are not required to) be induced by a  breaking point within the function $f_c(q)$ of some candidate $c$ with $f_c(q) = f(q)$ for $q\in [q'-\epsilon,q' + \epsilon]$ for some $\epsilon >0$. As a consequence, the slope of $f(q)$ can decrease at $q'$. On the other hand, breaking points of type $(ii)$ are guaranteed to be induced by a change of a candidate attaining the maximum budget. Hence, the slope at such a breaking point increases.

Reindex the agents according to their budget, i.e., $b_{a_1} \leq b_{a_2} \dots \leq b_{a_n}$.
Then, within 
each interval $[b_{a_i},b_{a_{i+1}}]$ there can only be breaking points of type $(ii)$. In particular, the function $f(q)$ is convex within these subintervals and its slope can take at most $n+1$ distinct values. 
In order to find~$q^*$, we now evaluate $f(q)$ at the borders of all of the intervals $[b_{a_i},b_{a_{i+1}}]$ for all $i \in [n-1]$ and select the left-most interval with $f(b_{a_{i^*}}) \leq 1 \leq f(b_{a_{i^*+1}})$. Then, we apply the Eisner-Severance method or its modified version as described in the proof of \Cref{thm:SeqPhragmen} in order to find the smallest $q^*\in [b_{a_{i^*}},b_{a_{i^*+1}}]$ such that $f(q^*)=1$. 

Both our preprocessing step and the Eisner-Severance method can be performed in $\mathcal{O}(n \cdot r_{waw})$. Doing so for all $k$ iterations yields an overall running time of $\mathcal{O}(kn \cdot r_{waw})$.   
\end{proof}

\subsection{Non-Sequential Thiele Rules in General Matching Elections}

In this section, we show that finding a winning committee in a general matching election is NP-hard for most $w$-Thiele rules. By contrast, as shown in the next two sections, this task becomes polynomial-time solvable for bipartite or symmetric matching elections. 

In the party-approval setting, computing a winning committee of non-constant size under PAV is NP-hard \cite{BGP+19a}. 
However, if $k$ is constant, the task can be solved in polynomial-time by iterating over all size-$k$ committees. 
This is in contrast to our setting, where we prove NP-hardness of computing a winning committee under a large class of $w$-Thiele rules including PAV, even for $k=2$. We reduce from the problem of deciding whether a 3-regular graph admits two edge-disjoint perfect matchings~\cite{Holye81a}.

\begin{restatable}[\appsymb]{theorem}{pavhard} \label{th:PAV-hard}
     Let $w$ be a weight sequence with $w_1> w_2>0$. Given a matching election $(N,A,k)$ and some number $\alpha\in \mathbb{R}$, deciding whether there exists a committee $\mathcal{M}$ of size $k$ with $\score_w(\mathcal{M})\geq \alpha$ is NP-complete for $k=2$ and even if each agent approves at most three agents. 
\end{restatable}

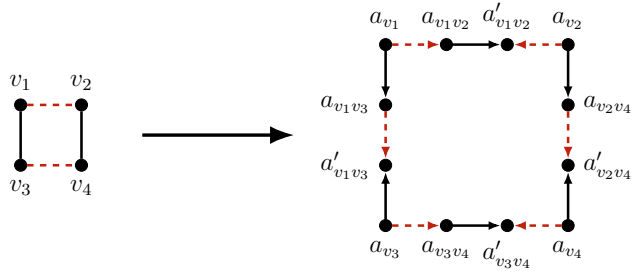
\begin{figure}
    \centering
    \scalebox{0.8}{
    \begin{tikzpicture}
         \node[node, label=90:$v_1$] (v1) at (-2, 0) {};
         \node[node, label=90:$v_2$] (v2) at (-1, 0) {};
         \node[node, label=270:$v_3$] (v3) at (-2, -1) {};
         \node[node, label=270:$v_4$] (v4) at (-1, -1) {};
         \draw (v1) edge[very thick,color1,dashed] (v2);
         \draw (v1) edge[very thick] (v3);
         \draw (v2) edge[very thick] (v4);
         \draw (v3) edge[very thick, color1, dashed] (v4);
         
         \node[node, label=90:$a_{v_1}$] (av1) at (4, 1) {};
         \node[node, label=90:$a_{v_2}$] (av2) at (7, 1) {};
         \node[node, label=270:$a_{v_3}$] (av3) at (4, -2) {};
         \node[node, label=270:$a_{v_4}$] (av4) at (7, -2) {};
         \node[node, label=180:$a_{v_1v_3}$] (av1v3) at (4, 0) {};
         \node[node, label=180:$a'_{v_1v_3}$] (aav1v3) at (4, -1) {};
         \node[node, label=90:$a_{v_1v_2}$] (av1v2) at (5, 1) {};
         \node[node, label=90:$a'_{v_1v_2}$] (aav1v2) at (6, 1) {};
         \node[node, label=0:$a_{v_2v_4}$] (av2v4) at (7, 0) {};
         \node[node, label=0:$a'_{v_2v_4}$] (aav2v4) at (7, -1) {};
         \node[node, label=270:$a_{v_3v_4}$] (av3v4) at (5, -2) {};
         \node[node, label=270:$a'_{v_3v_4}$] (aav3v4) at (6, -2) {};
         \draw[-{Latex[length=2mm]},very thick] (av1) -- (av1v3);
         \draw[-{Latex[length=2mm]},color1,very thick,dashed] (av1) -- (av1v2);
         \draw[-{Latex[length=2mm]},color1,very thick,dashed] (av2) -- (aav1v2);
         \draw[-{Latex[length=2mm]},very thick] (av4) -- (aav2v4);
         \draw[-{Latex[length=2mm]},color1,very thick,dashed] (av4) -- (aav3v4);
         \draw[-{Latex[length=2mm]},very thick] (av2) -- (av2v4);
         \draw[-{Latex[length=2mm]},color1,very thick,dashed] (av3) -- (av3v4);
         \draw[-{Latex[length=2mm]},very thick] (av3) -- (aav1v3);
         \draw[-{Latex[length=2mm]},color1,very thick,dashed] (av1v3) -- (aav1v3);
         \draw[-{Latex[length=2mm]},very thick] (av1v2) -- (aav1v2);
         \draw[-{Latex[length=2mm]},color1,very thick,dashed] (av2v4) -- (aav2v4);
         \draw[-{Latex[length=2mm]},very thick] (av3v4) -- (aav3v4);
         \draw[-{Latex[length=4mm]},ultra thick] (0,-0.5) -- (2.5,-0.5);
    \end{tikzpicture}
    }
    \caption{Example for the reduction from \Cref{th:PAV-hard}. The left side shows parts of a $3$-regular graph and the right side the constructed matching election. Red dashed arcs indicate how a matching in the graph is transformed to a matching in the  matching election.} 
    \label{fig:PAVred}
\end{figure}

\begin{proof}
We reduce from the problem of deciding whether a $3$-regular graph $G=(V,E)$ contains two edge-disjoint perfect matchings $M_1$ and $M_2$.\footnote{As this problem is equivalent to deciding whether $G$ is $3$-edge-colorable, NP-hardness follows from the work of \mbox{\citet{Holye81a}}.} 
Let $V=\{v_1,\dots, v_\eta\}$ and $E=\{e_1,\dots, e_m\}$ and observe that $3$-regularity of $G$ implies that $\eta$ is even and $m=3\eta/2$. 
From $G$, we construct a matching election $(N,A,k)$ as follows.  We introduce one \emph{node agent} $a_i$ for each $v_i\in V$. Moreover, for each edge $\{v_i,v_j\}\in E$ with $i<j$, we add an edge gadget consisting of one \emph{happy edge agent} $a_{ij}$ and one \emph{sad edge agent} $a'_{ij}$, where the node agent $a_{i}$ approves the happy edge agent $a_{ij}$, the happy edge agent $a_{ij}$ approves the sad edge agent $a'_{ij}$, and the node agent $a_{j}$ approves the sad edge agent $a'_{ij}$. We set $k=2$ and $\alpha=(5/2w_1 + 3/2w_2)\eta$ and refer to the two matchings to be found as $M'_1$ and $M'_2$.

We call a matching $M'$ of the agents $N$ a \emph{proper matching} if $M'$ is approved by all node agents and, for each edge $\{v_i,v_j\}\in E$ with $i<j$, it either holds that $\{a_{ij},a'_{ij}\}\in M'$ or that
both $\{a_{i},a_{ij}\} \in M'$ and $\{a'_{ij},a_{j}\}\in M'$:
A proper matching $M'$ matches all $\eta$ node agents to $\frac{\eta}{2}$ happy and $\frac{\eta}{2}$ sad edge agents and the remaining $m-\frac{\eta}{2}$ sad and $m-\frac{\eta}{2}$ happy edge agents to each other.  We show in the appendix that every two matchings $M'_1$ and $M'_2$ with $\score_w(\{M'_1, M'_2\})\geq \alpha$ 
need to be proper matchings. There exists a one-to-one correspondence between perfect matchings $M$ in $G$ and proper matchings $M'$ of the agents $N$ by including an edge $\{v_i,v_j\}\in E$ in $M$ if and only if $\{a_{ij},a'_{ij}\}\notin M'$. A visualization of the construction is depicted in \Cref{fig:PAVred}. One example for the described correspondence are the two matchings marked by dashed red edges. 

We now compute $\score_w(\{M'_1, M'_2\})$ assuming that $M'_1$ and $M'_2$ are proper matchings.
It is possible to calculate $\score_w(\{M'_1, M'_2\})$ by summing up the $w$-Thiele score of $M'_1$, i.e., $\score_w(\{M'_1\})$, and the marginal contribution of~$M'_2$ given~$M'_1$, i.e, $\score_w(\{M'_1, M'_2\})-\score_w(\{M'_1\})$. As we have assumed that $M'_1$ is a proper matching, it is approved by all node agents and by all happy edge agents not matched to node agents. Thus, $\score_w(M'_1)=(\eta+m-\frac{\eta}{2})w_1=2\eta w_1$, as the graph is $3$-regular. Turning to the marginal contribution of $M'_2$, as all node agents approve both matchings, they contribute $\eta w_2$. For the happy edge agents $a_{ij}$, it is possible to distinguish four different cases: 
\begin{description}
\setlength\itemsep{-0.25em}
\item[Case 1:] $\{a_{ij},a'_{ij}\} \not\in M'_1$ and $\{a_{ij},a'_{ij}\} \in M'_2$. In this case, the marginal contribution of $a_{ij}$ is $w_1$.
\item[Case 2:] $\{a_{ij},a'_{ij}\}\in M'_1$ and $\{a_{ij},a'_{ij}\} \in M'_2$. In this case, the marginal contribution of $a_{ij}$ is $w_2$.
\item[Case 3:] $\{a_{ij},a'_{ij}\}\in M'_1$ and $\{a_{ij},a'_{ij}\} \not\in M'_2$. In this case, the marginal contribution of $a_{ij}$ is $0$.
\item[Case 4:] $\{a_{ij},a'_{ij}\} \not\in M'_1$ and $\{a_{ij},a'_{ij}\} \not\in M'_2$. In this case, the marginal contribution of $a_{ij}$ is $0$.
\end{description}

By the assumption that $M'_1$ is proper, exactly $\frac{\eta}{2}$ happy edge agents are matched to node agents in~$M'_1$. Thus, Case 1 can occur at most $\frac{\eta}{2}$ times. Moreover, as there exist $\frac{3}{2}\eta$ happy edge agents and $\frac{\eta}{2}$ of them need to be matched to node agents in $M'_2$, Cases~1 and~2 combined can occur at most $\eta$ times. Thus, as $w_1> w_2$, the marginal contribution of $M'_2$ can be upper bounded by $\frac{\eta}{2}\cdot w_1 + \frac{\eta}{2} \cdot w_2$, leading to an upper bound for the combined score of any two proper matchings of $(5/2w_1 + 3/2w_2)\eta$.  Note that we set $\alpha$ exactly to match this upper bound and that it is only possible to achieve it if the second matching is chosen in a way such that Case 1 occurs $\frac{\eta}{2}$ times. We are now ready to show that there exist two edge-disjoint perfect matchings in $G$ if and only there exists two matching $M'_1$ and $M'_2$ with $\score_w(\{M'_1,M'_2\})\geq (5/2w_1 + 3/2w_2)\eta$ in the constructed matching election. 

$``\Rightarrow"$ Let $M_1$ and $M_2$ be two disjoint perfect matchings in $G$. Let $M'_1$ and $M'_2$ be the corresponding proper matchings of agents from $N$, i.e., for $t\in \{1,2\}$: 
\begin{align*}
    M'_t := & \bigcup_{\{v_i,v_j\} \in M_t: i<j} \{\{a_{i},a_{i j}\},\{a'_{i j},a_{j}\}\} \quad \cup \quad \bigcup_{\{v_i,v_j\} \in E\setminus M_t: i<j} \{\{a_{i j},a'_{i j}\}\}.
\end{align*}
As $M_1$ is a perfect matching, $M_1'$ has a score of $2\eta w_1$. Moreover, as $M_2$ is perfect and edge-disjoint from $M_1$, the first three cases concerning the marginal distribution of $M'_2$ from above all occur exactly $\frac{\eta}{2}$ times, while the last case does not appear at all. Hence, $\score_w(\{M'_1,M'_2\})=(5/2w_1 + 3/2w_2)\eta$. 

$``\Leftarrow"$ Let $M'_1$ and $M'_2$ be two matchings such that $\score_w(\{M'_1,M'_2\})\geq (5/2w_1 + 3/2w_2)\eta$, which implies that they are both proper matchings (a proof of this can be found in the appendix). Let $M_1$ and $M_2$ be the corresponding perfect matchings in $G$, i.e., for $t\in \{1,2\}$: 
\[M_t := \{\{v_i,v_j\} \in E \mid \{a_{i},a_{ij}\},\{a'_{i j},a_{j}\} \in M'_t\}\]

From $\score_w(\{M'_1,M'_2\})\geq (5/2w_1 + 3/2w_2)\eta$, it follows by our previous observations and as $w_1>w_2$ that Case 1 appears exactly $\frac{\eta}{2}$ times. Therefore, there exist $\frac{\eta}{2}$ happy edge agents that approve $M'_1$ but not $M'_2$. This implies that the corresponding edges are not included in $M_1$ but included in $M_2$. Thus, $M_1$ and $M_2$ are edge-disjoint.
\end{proof}

\subsection{Non-Sequential Thiele Rules in Bipartite Matching Elections} 
In contrast to the computational hardness even for $k=2$, all $w$-Thiele rules are tractable in bipartite matching elections. The general idea of the algorithm is to construct all $k$ matchings simultaneously with the help of a meta-election. In the meta-election, each agent is replaced by $k$ copies. We solve the \waw problem for this election with appropriate agent weights to obtain a single matching which matches all $k$ copies of each agent. From this, using Hall's theorem~\cite{Hall35a}, we construct $k$ matchings in the original instance. 
\begin{theorem} \label{thm:wThiele-bipartite}
Let $w$ be a weight sequence. In a bipartite matching election $(N,A,k)$, a winning committee under $w$-Thiele can be computed in $\mathcal{O}\big((kn)^3\big)$-time. 
\end{theorem}
\begin{proof}
Let $(N=N_1\dot{\cup}N_2,A,k)$ be a bipartite matching election and $w$ a weight sequence. We assume without loss of generality that $|N_1| = |N_2|$. If this is not the case, we add $||N_1|-|N_2||$ dummy agents to the smaller side, which are neither approved by any of the original agents nor approve any of them. Clearly, every matching in this new instance can be mapped to a matching in the original instance of equal $w$-Thiele score, and vice versa. 

We reduce our problem to solving one \waw instance of a meta matching election. In the meta-instance, we introduce $k$ copies $a^{(1)},\dots,a^{(k)}$ for every agent $a\in N$ and the weight of the $i$th copy is $w_i$, i.e., $\omega(a^{(i)})=w_i$ for all $i \in [k]$. 
An agent $a^{(i)}$ approves some $b^{(j)}$ iff $a$ approves $b$ in the original instance. Consider some outcome  $\widetilde{M}$ of the resulting \waw instance. Because of the special structure of the meta-instance and the fact that the weight sequence $w$ is non-increasing, we can assume without loss of generality that, for every agent $a$, there exists a threshold $i_a \in [k]$ such that her first $i_a$ copies are exactly those that are matched to partners she approves in $\widetilde{M}$. Hence, the contribution of the copies of an agent $a\in N$ to the weight of $\widetilde{M}$ under $\omega$ equals $\sum_{j = 1}^{i_a}w_j$. In the original instance, this is exactly the contribution of an agent to the $w$-Thiele score of a committee if she approves $i_a$ of the matchings in the committee. In the following, we show that, indeed, we can find a committee $\mathcal{M}$ of $k$ matchings in the original instance such that every agent $a\in N$ approves $i_a$ matchings in $\mathcal{M}$, i.e., $h_a(\mathcal{M}) = i_a$. 

In order to do so, we extend the matching $\widetilde{M}$  to a perfect matching in the meta-instance respecting the bipartition. Recall that we can do so since we assumed that $|N_1|=|N_2|$. From that, we construct a ``small'' bipartite graph $G$ which may contain parallel edges. More precisely, we define the multiset of edges $R$ of $G$ in the following, straightforward way: For every edge $\{a^{(i)},b^{(j)}\} \in \widetilde{M}, i,j \in [k]$, add one copy of the edge $\{a,b\}$ to $R$. Then, the multiset $R$ induces a bipartite graph $G=(N_1\dot{\cup}N_2,R)$ which is in particular $k$-regular. We extract $k$ perfect matchings from $G$ by a simply greedy procedure: With the help of Hall's Theorem \cite{Hall35a} it can be shown that every regular bipartite graph contains a perfect matching. We start by selecting some perfect matching in $G$ and set it to be $M_1$. Subsequently, we delete the edges contained in $M_1$ from $G$. Again, the obtained graph is regular and hence contains a perfect matching. By induction, we can proceed until we have selected $k$ perfect matchings. Lastly, we modify the extracted perfect matchings by deleting pairs that are not approved by any of the two endpoints, or in other words, making them minimal. Note that Pareto optimality of the constructed matchings is guaranteed by the Pareto optimality of $\widetilde{M}$.

Let $\mathcal{M} = \{M_1, \dots, M_k\}$ be the committee obtained from the above procedure. By construction, the number of matchings that are approved by some agent $a$ is exactly $i_a$ and hence $\score_{w}(\mathcal{M}) = \sum_{a \in N}\sum_{j=1}^{ h_a(\mathcal{M})}w_j = \sum_{a \in N}\sum_{j=1}^{i_a}w_j$. It remains to be shown that this is also optimal. Assume for contradiction that there exists a committee $\mathcal{M}'=\{M'_1, \dots, M'_k\}$ with $\score_{w}(\mathcal{M}') > \score_{w}(\mathcal{M})$. From $\mathcal{M}'$, we construct a matching $\widetilde{M}'$ in the meta-instance as follows. For all $i \in [k]$ and every $\{a,b\} \in M'_i$, add the pair $\{a^{(i)},b^{(i)}\}$ to the matching $\widetilde{M}'$. Now, for every agent $a$, it holds that the number of partners that its copies $a^{(1)}, \dots, a^{(k)}$ approve in the meta-instance matching $\widetilde{M}'$ equals $h_a(\mathcal{M}')$. However, so far, it is not guaranteed that the satisfied copies of $a$ are a prefix of $a^{(1)}, \dots, a^{(k)}$, which we need to ensure that the weight of $\widetilde{M}'$ under $\omega$ is maximal. We can ensure this by a simple exchange argument: Whenever there exists a copy $a^{(i)}$, matched to an unapproved agent, say $b^{(j)}$, while there exists another copy $a^{(i')}$ with $i'>i$ matched to an approved agent, say $c^{(j')}$, we replace the pairs $\{a^{(i)},b^{(j)}\}$ and $\{a^{(i')},c^{(j')}\}$ by the pairs $\{a^{(i)},c^{(j')}\}$ and $\{a^{(i')},b^{(j)}\}$. After doing this exhaustively, the contribution of the copies of any agent $a\in N$ to the weight of the matching $\widetilde{M}'$ under $\omega$ is exactly $\sum_{i = 1}^{h_a(\mathcal{M}')}w_i$. Hence, under~$\omega$, the weight of the constructed matching $\widetilde{M}'$ is $\score_w(\mathcal{M}')$, which is strictly larger than the weight of the matching $\widetilde{M}$, a contradiction. 

We conclude the proof by showing the claimed running time. Solving an instance of \waw in the meta-instance can be done in $\mathcal{O}((kn)^3)$-time. Computing $k$ perfect matchings in the ``small'' bipartite graph can be done in $\mathcal{O}((kn)^2)$-time. Lastly, transforming the resulting matchings to minimal matchings can be done in $\mathcal{O}(kn^2)$-time.
In total, we obtain a running time of $\mathcal{O}((kn)^3)$. 
\end{proof}

\subsection{Non-Sequential Thiele Rules in Symmetric Matching Elections} \label{sub:biptosym} \label{sub:biptosym} 
Unfortunately, the algorithm from the proof of \Cref{thm:wThiele-bipartite} does not work directly for symmetric matching elections, as not every (non-bipartite) $k$-regular graph can be partitioned into $k$ perfect matchings.  Nevertheless, it is still possible to extend the algorithm by reducing each symmetric matching election to an essentially equivalent bipartite matching election. 

Recall from \Cref{ob:equalapproval} that Pareto optimal matchings in symmetric matching elections have a strong structure, as they are, in particular, maximum matchings in the (undirected) approval graph. Using this, we can apply the 
Gallai-Edmonds Structure Theorem \citep{Gall64a,Edmo65a} (as stated in \Cref{ap:cc}) to obtain a partition of the agents into three sets $W,X,$ and $Y$ such that all agents from $X$ and $W$ approve every Pareto optimal matching. Moreover, in every Pareto optimal matching, all agents from $X$ are matched to agents from $Y$ and agents from $W$ are matched among themselves.
Using this, it is possible to transform every symmetric matching election into a bipartite one by putting agents from $Y$ on the one side and agents from $X$ and some dummy agents on the other side. It is then possible to construct from each winning committee under $w$-Thiele in the constructed bipartite election, a winning committee under $w$-Thiele in the original symmetric election. This is captured in the following lemma: 

\begin{restatable}[\appsymb]{lemma}{bipsymm} \label{lem:bip-symm}
    There exists a function $\psi$ mapping every symmetric matching election $(N,A,k)$ to a bipartite matching election $\psi\big((N,A,k)\big)$ and a function $\varphi$ mapping every committee in $\psi\big((N,A,k)\big)$ to a committee in $(N,A,k)$ such that, if a committee $\mathcal{M}$ is winning under $w$-Thiele in $\psi\big((N,A,k)\big)$, then $\varphi(\mathcal{M})$ is winning under $w$-Thiele in $(N,A,k)$. Both $\psi$ and $\varphi$ can be computed in $\mathcal{O}(n^3)$-time.
\end{restatable}

Using this lemma, we can extend the algorithm from \Cref{thm:wThiele-bipartite} to symmetric instances: 
\begin{corollary} \label{co:wThieleSym}
Let $w$ be a weight sequence. In a symmetric matching election $(N,A,k)$, a winning committee under $w$-Thiele can be computed in $\mathcal{O}((kn)^3)$-time.
\end{corollary}
\begin{proof}
Let $(N,A,k)$ be a symmetric matching election and $w$ be a weight sequence. 
\Cref{lem:bip-symm} implies that we can find a size-$k$ committee that is winning under $w$-Thiele by computing a size-$k$ committee that is winning under $w$-Thiele in the corresponding bipartite instance $\psi\big((N,A,k)\big)$ with the help of \Cref{thm:wThiele-bipartite} and then using the transformation~$\varphi(\cdot)$. 
\end{proof}

\section{Axiomatic Results}\label{sec:axioms}

\begin{table}
    \centering
    \scalebox{0.85}{
\begin{tabular}{llll}
\topline
\rowcolor{tableheadcolor}
\textbf{Rules} & Party-Approval Elections &  Symmetric Matching Elections \\
\midtopline
PAV & core stability \cite{BGP+19a} & core stability \\
seq-$w$-Thiele & not PJR \cite{ABC+16a}   & EJR (Theorem \ref{thm:SeqPAV}), not core stability (Proposition \ref{pr:Seq-core}) \\
Rule~X & EJR, not core stability \cite{PeSk20a}  & EJR, not core stability (Proposition \ref{pr:ruleX-core})\\
seq-Phragmén & PJR, not EJR \cite{BFJL16a}   & PJR, not EJR (Proposition \ref{pr:ruleX-core})\\
\bottomrule 
\end{tabular}
}
 \caption{Summary of results on the axiomatic properties of several multiwinner voting rules. If a rule satisfies an axiom in the party-approval setting, then this also holds in the setting of (symmetric) matching elections. Our positive result in Theorem \ref{thm:SeqPAV} holds for seq-$w$-Thiele rules satisfying $w_i >w_{i+1}$ for all $i \in \mathbb{N}$. 
 }
    \label{fig:ov2}
\end{table}

As matching elections are also party-approval elections, axiomatic guarantees from the latter setting still apply here, that is, 
PAV satisfies core stability, 
Rule~X satisfies EJR, and
seq-Phragmén satisfies PJR. 
 Below, we study whether stronger axiomatic guarantees are obtainable for our subdomain (see \Cref{fig:ov2} for an overview of our results). We focus on symmetric matching elections, as they exhibit a particularly strong structure. We start with a surprising positive result: A large class of sequential $w$-Thiele rules (including seq-PAV, which fails all considered axioms in general) satisfy EJR.

\begin{theorem}\label{thm:SeqPAV}
Let $w$ be a weight sequence with $w_i>w_{i+1}$ for all $i \in \mathbb{N}$. Seq-$w$-Thiele satisfies EJR in all symmetric matching elections.
\end{theorem}

\begin{proof}
Let $(N,A,k)$ be a symmetric matching election.
In \Cref{sub:firstOb} we have observed that the set $N$ of agents can be partitioned into three sets $W$, $X$, and $Y$, such that in any Pareto optimal matching, all agents in $W \cup X$ are matched, agents in $X$ are matched to agents in $Y$, and agents in $W$ are matched among themselves (see also \Cref{ap:cc}). Thus, a group of agents violating EJR can only contain agents from~$Y$. 

Let $\mathcal{M}=\{M_1, \dots, M_k\}$ be some output of seq-$w$-Thiele (we assume that seq-$w$-Thiele selected matching $i$ in iteration $i$). Let $\mathcal{M}_{< i}=\{M_1, \dots,M_{i-1}\}$. Assume for contradiction that there exists an EJR violation, i.e., for some $\ell\in [k]$, there is a set $S\subseteq N$ with $|S|\geq \ell n/k$, a Pareto optimal matching~$\widetilde{M}$ with $S \subseteq N_{\widetilde{M}}$ and $h_a(\mathcal{M}) < \ell$ for all $a \in S$.  

We claim that the existence of $S$ implies that in every iteration $i$, at least $|S|$ agents in $Y$ which are matched in this iteration approve at most $\ell -1$ matchings from $\mathcal{M}_{<i}$:

\begin{claim*}
For every $i \in [k]$, there exists a group $S_i \subseteq Y \cap N_{M_i}$ with $|S_i|=|S|$ and $h_a(\mathcal{M}_{< i}) \leq \ell - 1$ for all $a \in S_i$.
\end{claim*}

\begin{claimproof} 
Fix $i \in [k]$. If all agents in $S$ are matched in~$M_i$, the claim trivially holds when setting $S_i = S$. 
Consider some $a \in S$ which is not matched in $M_i$. Since $M_i$ and $\widetilde{M}$ are maximum matchings in the approval graph of the instance, their symmetric difference consists of alternating cycles and even-length paths. In particular, there exists an even-length path starting in $a$ and ending in some $b \in Y$ which is matched in $M_i$ but not in $\widetilde{M}$. If $h_b(\mathcal{M}_{< i}) > h_a(\mathcal{M}_{< i})$, we could strictly increase the marginal contribution of $M_i$ by augmenting along this path, as this would lead to $a$ approving $M_i$ at the cost of $b$ disapproving it. Hence, $h_b(\mathcal{M}_{<i}) \leq h_a(\mathcal{M}_{< i})$. 
Since all even-length paths in the symmetric difference of $M_i$ and $\widetilde{M}$ are disjoint, we can construct $S_i$ as follows: For every $a \in S$ choose $a$ itself if $a \in N_{M_i}$ and else the agent at the other end of the corresponding even-length alternating path. 
\end{claimproof}

Let $\mathcal{S}$ be the multiset of groups of agents $S_i$ from the claim, i.e., $\mathcal{S} := \{S_1, \dots, S_k\}$. We define $g_a(\mathcal{S}):=|\{i \in [k] \mid a \in S_i\}|$ as the number of sets in $\mathcal{S}$ that include agent~$a$. 
By construction, we know that $g_a(\mathcal{S})\leq \ell$ for all $a \in Y$: No group $S_i$ contains an agent that is already included in $\ell$ of the groups $S_1, \dots, S_{i-1}$, as this would imply that $a$ approves at least $\ell$ of the matchings in $\mathcal{M}_{<i}$.
Since $S_i \subseteq N_{M_i}$ for all $i \in [k]$, we have
$g_a(\mathcal{S}) \leq h_a(\mathcal{M}) \leq \ell - 1$ for all $a \in S$.
Moreover, $\sum_{a \in Y}g_a(\mathcal{S})=k|S|$, since every group~$S_i$ contains exactly $|S|$ agents from $Y$. We get
\begin{align*}
    \sum_{a \in Y}g_a(\mathcal{S}) & =  \sum_{a \in S}g_a(\mathcal{S}) +  \sum_{a \in Y\setminus S}g_a(\mathcal{S}) \\ 
    & \leq (\ell -1 ) |S| + \ell(|Y| - |S|) \\ 
    & = \ell |Y| - |S| \leq \frac{k|S||Y|}{n} - |S| < k|S|,
\end{align*} 
a contradiction (where the last step holds as $|Y| \leq n$).
\end{proof}

A natural follow-up question is whether sequential $w$-Thiele rules even satisfy the stronger axiom of core stability in symmetric matching elections. We prove that this is not the case.
\begin{restatable}[\appsymb]{proposition}{Seqcore} \label{pr:Seq-core}
    Let $w$ be a weight sequence. Committees returned by seq-$w$-Thiele are not guaranteed to be core stable, even if the given matching election is symmetric. 
\end{restatable} 

Furthermore, Rule~X and seq-Phragmén do not satisfy stronger guarantees in (symmetric) matching elections, compared to general party-approval elections.

\begin{restatable}[\appsymb]{proposition}{RuleXcore} \label{pr:ruleX-core}
     In symmetric matching election, committees returned by seq-Phragmén are not guaranteed to provide EJR and committees returned by Rule~X are not guaranteed to be core stable. 
\end{restatable}
\begin{proofS}[Proof (seq-Phragmén)]
    Consider a symmetric matching election consisting of three agents $a_1$, $a_2$, and $a_3$ all approving each other. We set $k=6$ and claim that the committee~$\mathcal{M}$ consisting of three times matching $\{\{a_1,a_2\}\}$ and three times matching $\{\{a_2,a_3\}\}$ is a winning committee under seq-Phragmén. In the first step, all possible non-empty matchings become affordable at $t=0.5$. Breaking ties, we select $\{\{a_1,a_2\}\}$. Now, $a_3$ has $0.5$ dollars left and thus needs to be included in the next matching.  Again breaking ties, we select $\{\{a_2,a_3\}\}$.
    Continuing this way of breaking ties, we alternate between adding $\{\{a_1,a_2\}\}$ and $\{\{a_2,a_3\}\}$ until $\mathcal{M}$ is constructed. However, $\mathcal{M}$ violates EJR, as the group $\{a_1,a_3\}$ is $4$-cohesive but $h_{a_1}(\mathcal{M})=h_{a_3}(\mathcal{M})=3$.
\end{proofS}

In the counterexamples for seq-$w$-Thiele, seq-Phragmén, and Rule~X, there also exist other winning committees under these rules that satisfy the respective notion. Presumably, this is due to the richness of the candidate space, combined with a high number of ties in the execution of all three rules. It remains an open question whether the rules always return at least one winning committee satisfying the respective property.

\section{Complexity of Checking Axioms}\label{sec:check-axiom}

In this section, we settle the computational complexity of checking whether a given committee provides a proportionality guarantee. We first consider EJR.  

Deciding whether a committee $W$ in a party-approval election provides EJR can be reduced to solving \waw: For each $\ell \in [k]$, we check whether there exists an $\ell$-cohesive group violating EJR by 
 marking all agents that approve less than $\ell$ matchings from $W$ and 
 checking whether there exists a candidate that is approved by at least $\ell \frac{n}{k}$ of the marked agents. The latter step can be solved by a single call to \waw by assigning to all marked agents a weight of one and to all other agents a weight of zero.

\begin{observation} \label{ob:check-EJR}
    Given a party-approval election $(N,C,A,k)$ and a committee $W$, it is possible to check whether $W$ provides EJR in $\mathcal{O}(k\cdot\runt)$-time.
\end{observation}

This approach does not extend to PJR. In fact, it turns out that checking whether a committee of a matching election provides PJR is coNP-complete. 
This is in contrast to general party-approval elections, for which this problem can be solved in polynomial time \citep{BGP+19a}.

\begin{restatable}[\appsymb]{proposition}{checkPJR}
    Given a matching election $(N,A,k)$ and a committee $\mathcal{M}$, checking whether~$\mathcal{M}$ provides PJR is coNP-complete, even if the given matching election is symmetric and bipartite.
\end{restatable}
\begin{proofS}
We reduce from the NP-hard \textsc{Clique} problem on $r$-regular graphs \cite{GJ79}, where given an undirected $r$-regular graph $G=(V,E)$ and an integer $q$, the question is whether there exists a set of $q$ pairwise adjacent nodes. 

We insert one \emph{node agent} $a_v$ for each node $v\in V$, $q$ \emph{dummy agents}, and $q$ \emph{good agents}. All node and dummy agents approve all good agents and the other way round. For each edge $\{u,v\}\in E$, we add a matching to $\mathcal{M}$ that matches $a_u$, $a_v$, and $q-2$ dummy agents to good agents. By adding agents with empty approval ballot and adding matchings that do not match node agents, we modify the instance such that $\frac{k}{n}=r-\frac{q-1}{2}+\frac{1}{q}$. Thus, each node agent approves~$r$ matchings from $\mathcal{M}$ and a group of $q$ agents deserves to be represented by $qr-\binom{q}{2}+1$ matchings. Intuitively, there exists a size-$q$ clique in $G$ if and only if $\mathcal{M}$ does not satisfy PJR, as for a group of node agents $X=\{a_v\mid v\in V'\}$, the set of matchings approved by some agent in~$X$ corresponds to the set of edges that are incident to some node in~$V'$. 
\end{proofS}

Note that, in our hardness reduction, the given committee has a non-constant size. In fact, given a party-approval election, the problem whether a committee $W$ provides PJR is solvable in $\mathcal{O}(2^{|W|}\cdot \runt)$-time: For all $\ell \in [k]$, we iterate over all $(\ell-1)$-subsets of candidates $W'\subseteq W$ and mark all agents whose approval set is a subset of $W'$. Subsequently, we check whether there exists a candidate approved by at least $\ell \frac{n}{k}$ of the marked agents.  In this case, the group of $\ell \frac{n}{k}$ agents is $\ell$-cohesive and by construction all of them approve only candidates from the set of $\ell-1$ candidates~$W'$.
\begin{observation}
    Given a party-approval election $(N,C,A,k)$ and a committee $W$, checking whether~$W$ provides PJR can be done in $\mathcal{O}(2^{|W|}\cdot \runt)$-time.
\end{observation}

Finally, we show that checking core stability is computationally intractable, even for a constant committee size.

\begin{restatable}[\appsymb]{proposition}{checkcore} \label{pr:check-core}
    Given a matching election $(N,A,k)$ and a committee $\mathcal{M}$, checking whether~$\mathcal{M}$ is core stable is coNP-hard, even if $k=6$
    and the given matching election is bipartite. 
\end{restatable}
\begin{proofS}
    We reduce from the NP-hard \textsc{Exact Cover by 3-Sets} (\textsc{X3C}) problem, where we are given a universe~$X$ of $3q$ elements and a collection $C$ of $3$-element subsets of~$X$ in which each element appears in exactly three sets and the question is whether there exists an exact cover $C'\subseteq C$ of $X$. An exact cover contains each element from $X$ exactly once.
    For each element $x\in X$, we insert one \emph{element agent} $a_x$ and one \emph{dummy element agent} $b_x$. Moreover, for each set $c\in C$, we add one \emph{set agent} $a_c$. For each element $x\in X$, the element agent $a_x$ and the dummy element agent~$b_x$ approve each other. Moreover, the element agent $a_x$ and the three set agents $a_c$ corresponding to sets in which $x$ is contained approve each other. Thus, the element agents form one side of the bipartition and the dummy element and set agents the other.
    We construct the committee $\mathcal{M}$ consisting of six matchings such that each dummy element agent approves one matching, each element agent approves two matchings, and each set agent approves two matchings. Moreover, we modify the instance such that each possible blocking coalition needs to deserve to be represented by three matchings and needs to contain all element and dummy element agents and $q$ set agents. It can be shown that a subset $C'\subseteq C$ is an exact cover if and only if the corresponding set agents together with all element and dummy element agents form a blocking coalition.
\end{proofS}

\section{Conclusion}

We initiated the study of a multiagent problem at the intersection of social choice and matching theory: Given preferences of agents over each other, we model the problem of finding a 
representative multiset of matchings 
as a multiwinner election. 
Notwithstanding the difficulty presented by an exponential candidate space, we exploit the structure of the election domain to recover the computational tractability of some considered sequential rules, 
and also establish computational and axiomatic results that do not hold in the general setting.  

There are several intriguing directions for future work on matching elections. 
First, 
one could consider
axioms that are tailored to the specific structure of the setting. For example, a natural relaxation of core stability could only allow groups of agents to be matched among themselves in a deviation. 
Second, it would be natural to allow agents to rank-order potential matching partners and apply ordinal multiwinner voting procedures. 
Third, it would be interesting to identify other relevant multiwinner voting domains involving compactly representable preferences over an exponential candidate space.

Finally, in some applications, one is interested in finding multiple matchings of the same set of agents to be implemented one after the other. It is therefore natural to try to find a \textit{sequence} of matchings, rather than simply a multiset (as done in this paper). While an arbitrary ordering of a proportional committee still provides proportionality if assessed as a whole, in such temporal settings, it might also be desirable to satisfy proportionality constraints for every sliding window of the sequence. 
One potential way to achieve this is to introduce depreciation weights to sequential rules, capturing the amount and recency of representation that agents have observed so far. Similar ideas have been recently explored  within the context of approval-based multiwinner elections~\citep{DBLP:conf/aaai/Lackner20}.

\section*{Acknowledgments}
This material is based on work supported by the Deutsche Forschungsgemeinschaft under grants NI~369/19 and BR~4744/2-1. 
We thank \'Agnes Cseh for helpful discussions.

\bibliographystyle{plainnat}
\bibliography{small,abb,algo}

\newpage
\appendix

\section{Omitted Proofs from Section \ref{sec:first_observation}}

\lemWeighApp*

\begin{proof}
Given a matching election $(N,A)$ and a weight function $\omega$ on the agents, let $G=(N,\A,E)$ be the corresponding approval graph. Recall that $G$ is a mixed graph, where $N$ is the set of nodes, $\A$ is the set of directed edges and $E$ is the set of undirected edges. We denote by $\bar{G}=(N,\bar{E})$ the undirected graph induced by $G$. More precisely, $\bar{E} = \{\{a,b\} \mid \{a,b\} \in E \text{ or } (a,b) \in \A \}$. We show how to solve the \textsc{Weighted Approval Winner} problem by computing two maximum weight matchings in $\bar{G}$, with respect to two different weight functions. We start by defining the first weight function on the edges $w:\bar{E} \rightarrow \mathbb{R}_{\geq 0}$. For every directed edge in $G$, $(a,b) \in \A$, let $w(\{a,b\}) = \omega_a$ and for every undirected edge in $G$, $\{a,b\} \in E$ , let $w(\{a,b\}) = \omega_a + \omega_b$.

By construction of the weight function $w$ it holds that for every matching $M$ in $\bar{G}$ the weight of $M$ with respect to $w$ equals the weighted sum of all agents under $\omega$ that approve $M$, that is, $\sum_{e \in M} w(e) = \sum_{a \in N_M} \omega_a$. Let $M$ be a maximum weight matching in $\bar{G}$ with respect to $w$. By the above observation, $M$ also maximizes the weighted approval sum under $\omega$ among all matchings of the agents. Recall that, in order for $M$ to be a candidate within the matching election $(N,A)$, it needs to be minimal and Pareto optimal. While $M$ clearly satisfies minimality (every edge included in $M$ is approved by at least one agent), Pareto optimality is not guaranteed since there might exist agents $a \in N$ with $\omega_a = 0$. 

In the following, we construct a matching $M'$ in $\bar{G}$ based on the matching $M$, that is minimal, maximizes the weighted sum of approvals and is also guaranteed to be Pareto optimal. To this end we first define a second weight function on the agents, i.e., $\omega':N \rightarrow \mathbb{R}_{\geq 0}$. More precisely, $\omega'_a = n+1$ if $a \in N_M$ and $\omega'_a = 1$ if $a \not\in N_M$. Again, we derive a weight function $w':\bar{E} \rightarrow \mathbb{R}_{\geq 0}$ on the edges of $\bar{G}$ as follows. For every directed edge in $G$, $(a,b) \in \A$, let $w'(\{a,b\}) = \omega'_a$ and for every undirected edge in $G$, $\{a,b\} \in E$ , let $w'(\{a,b\}) = \omega'_a + \omega'_b$. Again, by construction of $w'$ it holds for every matching $M'$ in $\bar{G}$ that $\sum_{e \in M'} w'(e) = \sum_{a \in N_{M'}} \omega'_a$.

Let $M'$ be a maximum weight matching in $\bar{G}$ with respect to $w'$. We first claim that $N_M \subseteq N_{M'}$. For the sake of contradiction, assume that this is not the case. Then, \[\sum_{e \in M} w'(e) = \sum_{a \in N_M} \omega'_a = |N_M| \cdot (n+1) > (|N_M|-1) \cdot (n+1) + n \geq \sum_{a \in N_{M'}} \omega'_a = \sum_{e \in M'}w'(e),\] which contradicts the maximality of $M'$ with respect to $w'$. Hence, \[\sum_{e \in M} w(e) = \sum_{a \in N_M} \omega_a \leq \sum_{a \in N_{M'}} \omega_a = \sum_{e \in M'} w(e), \] and by the maximality of $M$ with respect to $w$ the two sides are equal. Hence, $M'$ also maximizes the weighted approval sum with respect to $\omega$. Moreover, $M'$ is minimal, since every edge in $M'$ is approved by at least one agent. Lastly, it remains to show that $M'$ is also Pareto optimal. Assume for contradiction that there exists $M''$ with $N_{M'} \subsetneq N_{M''}$. However, since $\omega'$ is strictly positive for all agents, this would imply \[\sum_{e \in M''} w'(e) = \sum_{a \in N_{M''}} \omega'_a > \sum_{a \in N_{M'}} \omega'_a = \sum_{e \in M'} w'(e),\] a contradiction to the maximality of $M'$ with respect to $w'$. We conclude that $M'$ is a \textsc{Weighted Approval Winner} for the matching election $(N,A)$ and the weight function $\omega$. 

Summarizing, we have shown that the \textsc{Weighted Approval Winner} problem for any matching election can be solved by computing two maximum weight matchings. This can be done in $\mathcal{O}(n^3)$ time (Theorem 11.19 in \cite{10.5555/2190621}).
\end{proof}
\section{Omitted Proofs from Section \ref{sec:cc}} \label{ap:cc}
\lpav*
\begin{proof}
\citet{BGP+19a} showed that for a party-approval election $(N,C,A,k)$ a core-stable committee can be computed by running a parameterized local search variant of PAV. The method was originally introduced by \citet{AEH+18a} for general approval-based multiwinner elections. In the following, we present their method tailored to the party-approval setting and show that computing a winning committee under it can be reduced to solving the \waw problem $\mathcal{O}(nk^3\ln(k))$ times. 

Let $w$ be the weight sequence corresponding to PAV, i.e., $w_i = 1/i$ for all $i \in \mathbb{N}$. The method \emph{LS-PAV} starts by selecting an arbitrary size-$k$ committee $W$. Then, it checks whether there exists an \emph{improving swap} defined as follows. A \emph{swap} replaces one candidate $c$ which occurs at least once in $W$ by some other candidate $c'\neq c$. Let $W'$ be the committee obtained from $W$ by removing (one copy of) $c$ and adding one copy of $c'$. The swap replacing $c$ by $c'$ is called \emph{improving} iff \[\score_w(W')\geq \score_w(W) + \epsilon,\] where $\epsilon:= \frac{1}{(1+2(k-1))(k-1)k}$. LS-PAV searches for an improving swap $(c,c')$ and, if an improving swap exists, updates the committee by exchanging (one copy of) $c$ for (one copy of) $c'$. This procedure is repeated until there does not exist any improving swap.

We claim that we can check whether there exists an improving swap (and if so find one) in $\mathcal{O}(k \cdot r_{waw})$-time: For a given committee $W$, iterate over all $c$ that are selected at least once in $W$. Define $\widetilde{W}$ as the committee obtained from $W$ by deleting (one copy of) $c$. We create a $\waw$ instance by setting the weights of the agents to $\omega(a) = w_{h_a(\widetilde{W}) + 1}$. Let $c'$ be a weighted approval winner of this instance. Then, there exists an improving swap replacing $c$ iff $(c,c')$ is an improving swap. Moreover, \citet{BGP+19a} showed that the algorithm always terminates after performing at most $\mathcal{O}(nk^3\ln(k))$ improving swaps and that the outcome is guaranteed to satisfy core stability. 
\end{proof}

\pavhard*
\begin{proof}[Proof (continued)]
It remains to be proven that for every two matching $M'_1$, $M'_2$ in the constructed matching election with $\score_w(\{M'_1,M'_2\})\geq(5/2w_1 + 3/2w_2)\eta$ it needs to hold that both of them are proper matchings. For every matching $M'$ of agents in $N$ and each edge $\{v_i,v_j\}\in E$ with $i<j$, one of the following four cases has to hold:
\begin{description}
\setlength\itemsep{-0.2em}
    \item[Case 1:] $\{a_{i}, a_{ij}\}$ and $\{a_{j}, a'_{ij}\}\in M'$
    \item[Case 2:] 
    \big($\{a_{i}, a_{ij}\}\in M'$ and $\{a_{j}, a'_{ij}\}\notin M'$\big) or  \big($\{a_{i}, a_{ij}\}\notin M'$ and $\{a_{j}, a'_{ij}\}\in M'$\big)
    \item[Case 3:] $\{a_{ij}, a'_{ij}\}\in M'$
    \item[Case 4:] None of the three edges is part of $M'$
\end{description}
 Note that the last case never occurs, as $M'$ cannot be Pareto optimal (we can additionally match the happy and sad edge agent which leads to a strict extension of $N_{M'}$). This implies that the first three cases together happen $\frac{3}{2}\eta$ times.
 Let $y$ denote the frequency of the first and $z$ the frequency of the second case in $M'_1$
 and $\tilde{y}$ and $\tilde{z}$ their frequencies in $M'_2$.
 
We now bound the score of $M'_1$ and the marginal contribution of $M'_2$ in these four variables. The number of agents that approve $M'_1$ is $2y$ plus $z$ plus the number of times the third case appears, which is $\frac{3}{2}\eta-y-z$:
\[\score_w(M'_1) 
=\left(\frac{3}{2}\eta+y\right)w_1 \text.\]
Turning to the marginal contribution of the second matching $M'_2$, we first consider the contribution of the node agents. We know that $M'_2$ is approved by $2\tilde{y}+\tilde{z}$ node agents. As $M'_1$ is approved by $2y+z$ node agents, at most $\eta-2y-z$ node agents can contribute with $w_1$ to the marginal score of $M'_2$, while the remaining $2\tilde{y}+\tilde{z}-(\eta-2y-z)$ contribute with $w_2$. Turning to the edge agents, $M'_2$ is approved by  $\frac{3}{2}\eta-\tilde{y}-\tilde{z}$ happy edge agents. As the first matching is approved by all but $y+z$ happy edge agents, the number of happy edge agents contributing $w_1$ to the marginal score of $M'_2$  can be upper bounded by $y+z$, while the remaining $\frac{3}{2}\eta-\tilde{y}-\tilde{z}-y-z$ happy edge agents approving $M'_2$ contribute $w_2$. Thus, the marginal contribution
of $M'_2$ can be upper bounded as:
\[\score_w(\{M'_1,M'_2\}) - \score_w(\{M'_1\}) \leq (\eta-y)w_1+(\frac{\eta}{2}+\tilde{y}+y)w_2 \text.\]
Combining the two bound yields:
\[\score_w(\{M'_1,M'_2\}) \le \frac{5}{2}\eta w_1+ (\frac{\eta}{2}+\tilde{y}+y)w_2. \] 
Recall that $y,\tilde{y}\leq \frac{\eta}{2}$. Thus, as we have assumed that $\score_w(\{M'_1,M'_2\})\geq(5/2w_1 + 3/2w_2)\eta$, it needs to hold that $y=\tilde{y}=\frac{\eta}{2}$, which implies that $z=\tilde{z}=0$. From this it directly follows that $M'_1$ and $M'_2$ are proper matchings.  
\end{proof}

As each Pareto optimal matching in a symmetric matching elections corresponds to a maximum matching in the corresponding approval graph, we make use of the Gallai-Edmonds decomposition \cite{Gall64a,Edmo65a} to transform a symmetric matching election into a bipartite matching election:

\paragraph{Gallai-Edmonds decomposition.} Let $G=(V,E)$ be an undirected graph and $W,X,$ and $Y$ be a partition of the set of nodes $V$, such that $Y$ is the set of nodes which are not matched in all maximum matchings, $X$ are their neighbors from $V \setminus Y$, and ${W = V\setminus (Y\cup X)}$. Concerning the notation, we denote by $G[S]$ the subgraph induced by $S \subseteq V$, i.e., the graph $(S,E[S])$, where $E[S]$ is the set of all edges from $E$ having both end nodes in $S$.  The decomposition theorem \cite{Gall64a,Edmo65a} says that 
\begin{enumerate}
\setlength\itemsep{-0.25em}
 \item the graph $G[W]$ contains a perfect matching; \label{item:Gall-Edm-1}
 \item the connected components of $G[Y]$ are all factor-critical, i.e., removing any node from a connected component of $G[Y]$ results in a graph containing a perfect matching; and \label{item:Gall-Edm-2}
 \item in every maximum matching, all nodes from $X$ are matched to distinct connected components of $G[Y]$. \label{item:Gall-Edm-3}
\end{enumerate}

Using the Gallai-Edmonds decomposition of the approval graph, we can prove the following lemma:
\bipsymm*
\begin{proof}
Let $(N,A,k)$ be a symmetric matching election. Applying the Gallai-Edmonds decomposition to the approval graph $G$ of $(N,A,k)$, we can partition the set of agents into three sets $W,X,Y$ with the above described properties. For any matching $M$ in $(N,A,k)$, the only relevant information to determine $N_M$ is the matching between the agents in $X$ and the agents in $Y$. Following this idea, we construct the bipartite and symmetric matching election $\psi\big((N,A,k)\big)=(N'=N'_1 \dot{\cup}N'_2,A',k)$ that provides this information as follows. We set $N'_1=Y$ and $N'_2 = X \cup D$, where $D$ is a set of dummy nodes. More precisely, $D$ is constructed as follows: Let $Y_1,\dots, Y_{\ell}$ be the subgroups of agents corresponding to the connected components in $G[Y]$. For some group $Y_i$, we add $|Y_i|-1$ dummy agents $d_i^{(1)},\dots, d_i^{(|Y_i|-1)}$ to $D$. These agents approve the agents of $Y_i$ and vice versa. Lastly, agents from $Y$ and $X$ approve each other in the new preference profile $A'$ iff they approve each other in the original preference profile $A$. 

We further define two transformations $\mu$ and $\varphi$ that, given a Pareto optimal matching in the symmetric instance $(N,A,k)$, return a Pareto optimal matching in $\psi\big((N,A,k)\big)=(N'=N'_1 \dot{\cup}N'_2,A',k)$ and vice versa. For a Pareto optimal matching $M$ in $(N,A,k)$, we define $\mu(M)$ as follows: For each pair between an agent from $X$ and an agent from $Y$ in $M$, we add the same pair to $\mu(M)$. For all groups $Y_i$ which now already have one matched agent, we match the remaining agents to their corresponding dummy agents. For all other groups $Y_i$, we leave exactly the agent unmatched which is unmatched in $M$ and match the remaining agents to their corresponding dummy agents. Observe that this transformation maintains the set of agents in $Y$ that are matched, i.e., $N_M \cap Y = N'_{\mu(M)} \cap Y$. For the opposite direction, $\varphi(\cdot)$, let $M'$ be a Pareto optimal matching in the constructed bipartite graph. By using Pareto optimality, it can be shown that all dummy agents need to be matched to agents they approve and at most one agent from each group $Y_i$ is matched to an agent from $X$. Moreover, all agents in $X$ need to be matched to agents they approve. We define the matching $\varphi(M')$ by first adding all pairs between agents from $X$ and agents from $Y$ in $M'$ and a perfect matching of the agents in $W$. Lastly, for all groups $Y_i$, we add a matching leaving exactly one agent in $Y_i$ unmatched. More precisely, for those groups having an agent matched to an agent from $X$, we leave this agent unmatched and for a group $Y_i$ not having any agent matched to an agent from $X$, we leave the same agent unmatched which is unmatched in $M'$. Similarly to before, this transformation maintains the set of agents in $Y$ that are matched, i.e., $N'_{M'} \cap Y = N_{\varphi(M')} \cap Y$. 

We straightforwardly extend the two transformations from matchings to committees of matchings. More precisely, for a given committee $\mathcal{S}=\{M_1, \dots, M_k\}$ in $(N,A,k)$, we define $\mu(\mathcal{S}) := \{\mu(M_1),\dots,\mu(M_k)\}$ and for a committee $\mathcal{M}=\{M'_1, \dots, M'_k\}$ in $\psi\big((N,A,k)\big)$, we define $\varphi(\mathcal{M}) := \{\varphi(M'_1),\dots,\varphi(M'_k)\}$. In order to clearly distinguish both instances, we write $h_a(\mathcal{S})$ for the number of matchings agent $a$ from instance $(N,A,k)$ approves in $\mathcal{S}$ and $h'_a(\mathcal{M})$ for the number of matchings in $\mathcal{M}$ an agent $a$ from instance $\psi\big((N,A,k)\big)$ approves. Observe that for some committee $\mathcal{S}$ in $(N,A,k)$ it holds that $h_a(\mathcal{S}) = h'_a(\mu(\mathcal{S}))$ for all $a \in Y$. Symmetrically, for some committee $\mathcal{M}$ in $\psi\big((N,A,k)\big)$ it holds that $h'_a(\mathcal{M}) = h_a(\varphi(\mathcal{M}))$ for all agents $a \in Y$. 

We now turn to proving that $\varphi$ and $\psi$ fulfill the property stated in the theorem. Assume for contradiction that $\mathcal{M}$ is winning under $w$-Thiele in $\psi\big((N,A,k)\big)$, but $\varphi(\mathcal{M})$ is not winning under $w$-Thiele in $(N,A,k)$. Hence, there exists a size-$k$ committee $\mathcal{S}$ in $(N,A,k)$ with $\score_w(\mathcal{S}) > \score_w(\varphi(\mathcal{M}))$. In particular, this implies that \begin{equation}
\sum_{a \in Y} \sum_{i=1}^{ h_a(\mathcal{S})}w_i > \sum_{a \in Y} \sum_{i=1}^{h_a(\varphi(\mathcal{M}))}w_i. \label{eq:one}\end{equation} Now, using one of our transformations, we can find a committee $\mu(\mathcal{S})$ in the bipartite instance such that all agents in $X\cup D$ are matched $k$ times and $h_a(\mathcal{S}) = h'_a(\mu(\mathcal{S}))$ for all agents in $a \in Y$, i.e., an agent from $Y$ in the bipartite instance approves the same number of matchings from $\mu(\mathcal{S})$ as the corresponding agent from the symmetric instance approves in $\mathcal{S}$. 
We get \vspace{-0.2cm}
\begin{align*}
    \score_w(\mu(\mathcal{S})) & = \sum_{a \in Y} \sum_{i=1}^{h'_a(\mu(\mathcal{S}))}w_i + \sum_{i=1}^k w_i \cdot (|X| + |D|)  = \sum_{a \in Y}\sum_{i=1}^{h_a(\mathcal{S})}w_i + \sum_{i=1}^k w_i \cdot (|X| + |D|) \\ 
    & > \sum_{a \in Y}\sum_{i=1}^{h_a(\varphi(\mathcal{M}))}w_i + \sum_{i=1}^k w_i \cdot (|X| + |D|)  = \sum_{a \in Y} \sum_{i=1}^{h'_a(\mathcal{M})}w_i + \sum_{i=1}^k w_i \cdot (|X| + |D|) \\ 
    & = \score_w(\mathcal{M}),
\end{align*}
where the inequality follows from (\ref{eq:one}). This yields a contradiction to the optimality of $\mathcal{M}$.

Concerning the running time of $\psi(\cdot)$, note that a Gallai-Edmonds decomposition can be computed by running Edmond's blossom algorithm \citep{Edmo65a} once which needs $\mathcal{O}(n^3)$-time. Given such a decomposition, constructing $\psi(\cdot)$ can be done in $\mathcal{O}(n^2)$-time.
On the other hand, applying the transformation $\varphi(\cdot)$, we have to compute one maximum cardinality matching of the vertices $Y_i$ for each $i\in [\ell]$.  Since the groups $Y_i$ correspond to the connected components of $G[Y]$, this can be done by computing one maximum cardinality matching in $G[Y]$ (where some nodes were deleted). This can be done in $\mathcal{O}(n^3)$-time. 
\end{proof}


\section{Omitted Proofs from Section \ref{sec:axioms}}
\Seqcore*
\begin{proof}

\tikzstyle{mnode}=[draw, circle, fill, inner sep = 0.5pt]
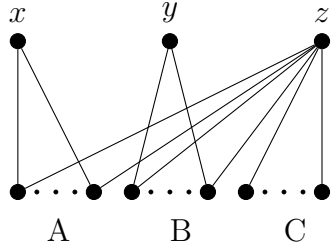
\begin{figure}
    \centering
    \begin{tikzpicture}
         \node[node, label=90:$z$] (v1) at (4, 0) {};
         \node[node, label=90:$y$] (v2) at (2, 0) {};
         \node[node, label=90:$x$] (v3) at (0, 0) {};
         
         \node[node] (z1) at (4,-2) {};
         \node[mnode]  at (3.75,-2) {};
         \node[mnode]  at (3.5,-2) {};
         \node[mnode]  at (3.25,-2) {};
         \node[node] (zz1) at (3,-2) {};
         \node[node] (z2) at (1.5,-2) {};
         \node[node] (zz2) at (2.5,-2) {};
         \node[mnode]  at (1.75,-2) {};
         \node[mnode]  at (2,-2) {};
         \node[mnode]  at (2.25,-2) {};
         \node[node] (z3) at (1,-2) {};
         \node[mnode]  at (0.75,-2) {};
         \node[mnode]  at (0.5,-2) {};
         \node[mnode]  at (0.25,-2) {};
         \node[node] (zz3) at (0,-2) {};
         \node[text width=1cm] at (0.9,-2.5) {A};
        \node[text width=1cm] at (2.5,-2.5) {B};
        \node[text width=1cm] at (4,-2.5) {C};
        \draw (v3) edge (z3);
         \draw (v3) edge (zz3);
         \draw (v2) edge (z2);
         \draw (v2) edge (zz2);
         \draw (v1) edge (z1);
         \draw (v1) edge (zz1);
         \draw (v1) edge (z2);
         \draw (v1) edge (zz2);
         \draw (v1) edge (z3);
         \draw (v1) edge (zz3);
    \end{tikzpicture}
    \caption{Approval graph of counterexample for core stability for sequential $w$-Thiele rules from \Cref{pr:Seq-core}.}
    \label{fi:Seq-core}
\end{figure}

To show the proposition, we present a symmetric matching election and construct a committee which is winning under seq-$w$-Thiele but fails to be core stable.\footnote{To make some of the calculations easier, we construct the instance in a way such that $n=k$. Thus, the example should not be understood as a minimal counterexample.}
The instance consists of three groups of \emph{dummy agents} $A=\{a_1,\dots,a_{27}\}$, $B=\{b_1,\dots,b_{27}\}$, and $C=\{c_1,\dots,c_{41}\}$ and three \emph{special agents} $x$, $y$, and $z$. Approvals are symmetric and the special agent $x$ approves all agents from $A$, the special agent $y$ approves all agents from $B$, and the special agent $z$ approves all dummy agents. See \Cref{fi:Seq-core} for a visualization. Note that this instance consists of $98$ agents. We set $k=n=98$. Thus, every agents deserves to be represented by one matching.  

We now construct a committee $\mathcal{M}$ that is winning under seq-$w$-Thiele and argue that it is not core stable. In the first nine matchings, we match $x$ and $z$ to distinct agents from $A$ and $y$ to distinct agents from $B$. In the matchings ten to eighteen, we match $y$ and $z$ to previously unmatched agents from $B$ and $x$ to a previously unmatched agent from $A$. Note that the selected matchings are winning under seq-$w$-Thiele in their respective round, as we match only so-far unmatched dummy agents and assume $w_1\geq w_2$. Overall, all agents from $A$ and $B$ are matched in exactly one of the first eighteen matchings. In the remaining $80$ matchings, we match $x$ to an agent from $A$, $y$ to an agent from $B$, and $z$ to an agent from $C$ such that approvals within $A$, $B$, and $C$ are distributed as equally as possible. We can do so by constructing the matchings sequentially and always matching each special agent to the so far unhappiest agent from the respective group. Note that it is possible to distribute the approvals as equally as possible within a set, as we have assumed that $w_i\geq w_{i+1}$ for all $i\in  \mathbb{N}$. Moreover, after matching eighteen, it is always possible to match $z$ to an agent of $C$ in a winning matching, as over the whole construction process, each node from $A$ and $B$ approves the same or more of the already added matchings than a node from $C$ ($|B|,|A|<|C|$).

To summarize, the summed happiness score of the agents from the three different sets are as follows: $\sum_{a\in A} h_a(\mathcal{M})=\sum_{a\in B} h_a(\mathcal{M})= 98+9=107$ and $\sum_{a\in C} h_a(\mathcal{M})=80.$ Note that it holds that $\frac{107+1}{27}=4$ and $\frac{80+2}{41}=2$. By the pigeonhole principle, this implies that there exists at least one agent $a$ from $A$ that approves only three matchings from $\mathcal{M}$, at least one agent $b$ from $B$ that approves only three matchings, and, as happiness scores are distributed as equally as possible, two agents $c$ and $c'$ from $C$ which only approve one matching. We claim that the group $\{a,b,c,c'\}$ blocks $\mathcal{M}$. Note that this group deserves to be represented by four matchings. Let $\mathcal{M'}$ be a set of four matchings, where $a$ is matched to $x$ and $b$ is matched to $y$ in all four matchings, while in two matchings, $c$ is matched to $z$ and in the other two, $c'$ is matched to $z$. As all four agents approve strictly more matchings from $\mathcal{M'}$  than from $\mathcal{M}$, core stability is violated. 
\end{proof}

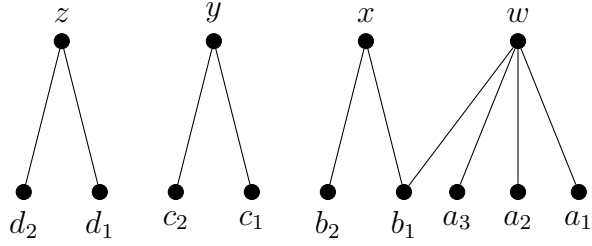
\begin{figure}
    \centering
    \begin{tikzpicture}
         \node[node, label=90:$z$] (v1) at (0, 0) {};
         \node[node, label=90:$y$] (v2) at (2, 0) {};
         \node[node, label=90:$x$] (v3) at (4, 0) {};
         \node[node, label=90:$w$] (v4) at (6, 0) {};

         \node[node,, label=270:$d_2$] (d1) at (-0.5,-2) {};
         \node[node,, label=270:$d_1$] (d2) at (0.5,-2) {};
         \node[node,, label=270:$c_2$] (c1) at (1.5,-2) {};
         \node[node,, label=270:$c_1$] (c2) at (2.5,-2) {};
         \node[node,, label=270:$b_2$] (b1) at (3.5,-2) {};
         \node[node,, label=270:$b_1$] (b2) at (4.5,-2) {};
         \node[node,, label=270:$a_3$] (a1) at (5.2,-2) {};
         \node[node,, label=270:$a_2$] (a2) at (6,-2) {};
         \node[node,, label=270:$a_1$] (a3) at (6.8,-2) {};
         
        \draw (v1) edge (d1);
        \draw (v1) edge (d2);
        \draw (v2) edge (c1);
        \draw (v2) edge (c2);
        \draw (v3) edge (b1);
        \draw (v3) edge (b2);
        \draw (v4) edge (a1);
        \draw (v4) edge (a2);
        \draw (v4) edge (a3);
        \draw (v4) edge (b2);
    \end{tikzpicture}
    \caption{Approval graph of counterexample for core stability for Rule~X from \Cref{pr:ruleX-core}.}
    \label{fi:ruleX-core}
\end{figure}
\RuleXcore*
\begin{proof}[Proof (Rule~X)]

    We depict our counterexample in \Cref{fi:ruleX-core}. It consists of $13$ agents $\{w$, $x$, $y$, $z$, $a_1$, $a_2$, $a_3$, $b_1$, $b_2$, $c_1$, $c_2$, $d_1$, $d_2\}$. Approvals are symmetric. Agent $w$ approves $a_1$, $a_2$, $a_3$, and $b_1$. Agent $x$ approves $b_1$ and $b_2$. Agent $y$ approves $c_1$ and $c_2$. Agent $z$ approves $d_1$ and $d_2$. We set $k=13$. Thereby, each agent starts with a budget of one dollar. We now describe a run of Rule~X on the constructed instance which returns a committee $\mathcal{M}$ that is not core stable. Initially, all matchings which are approved by eight agents are $\frac{1}{8}$-affordable. Breaking ties, we select the matching $\{\{w,b_1\},\{x,b_2\},\{y,c_1\},\{z,d_1\}\}$ eight times. After that, all agents except $a_1$, $a_2$, $a_3$, $c_2$, and $d_2$ have zero budget left. Now, every matching which is approved by one of $a_1$, $a_2$, $a_3$, and $c_2$, and $d_2$ is $\frac{1}{3}$-affordable. We select the matching $\{\{w,a_1\},\{x,b_2\},\{y,c_2\},\{z,d_2\}\}$ three times. Subsequently, only $a_2$ and $a_3$ have budget, which makes all matchings which are approved by one of them $1$-affordable. We select $\{\{w,a_2\},\{x,b_2\},\{y,c_1\},\{z,d_1\}\}$ and $\{\{w,a_3\},\{x,b_2\},\{y,c_1\},\{z,d_1\}\}$ as the last two matchings. Note that $a_2$ and $a_3$ both approve one matching from $\mathcal{M}$, while $c_2$ and $d_2$ approve three matchings from~$\mathcal{M}$. Let $\mathcal{M'}$ be a set of four matchings, where all matchings match $z$  to $d_2$ and $y$ to $c_2$, two of the matchings match $w$ to $a_2$ and the remaining two matchings match $w$ to $a_3$. The group $\{a_2,a_3,c_2,d_2\}$ block $\mathcal{M}$, as they deserve to be represented by four matchings and all four agents approve more matchings from $\mathcal{M'}$ than from $\mathcal{M}$.
    
    Note that, as the counterexample of \citet{PeSk20a} showing that Rule~X violates core stability only allows that every candidate can be selected once (and also partially relies on this constraint), this example also settles the question whether Rule~X satisfies core stability in every party-approval election. 
\end{proof}

\section{Omitted Proofs from Section \ref{sec:check-axiom}}
\checkPJR*
\begin{proof}
We reduce from the NP-hard \textsc{Clique} problem on $r$-regular graphs \cite{GJ79}, where given an undirected $r$-regular graph $G=(V,E)$ and an integer $q$ the question is whether there exists a set of $q$ pairwise adjacent nodes. We assume without loss of generality that $q>3$. We construct a matching election and a committee $\mathcal{M}$ as follows.

We insert one \emph{node agent} $a_v$ for each node $v\in V$, $q$ \emph{dummy agents}, and $q$ \emph{good agent}s. All node and dummy agents approve all good agents and the other way round. Turning to the construction of $\mathcal{M}$, for each edge $\{u,v\}\in E$, we add a matching to $\mathcal{M}$ that matches $a_u$, $a_v$, and $q-2$ dummy agents to good agents. Further, we insert $2|E|+1$ matchings in which each dummy agent is matched to a good agent. Lastly, we modify the instance such that $\frac{k}{n}=r-\frac{q-1}{2}+\frac{1}{q}$
by adding agents with empty approval ballot and matchings that match each dummy agent to a good agent. Note that each node agent approves $r$ matchings from $\mathcal{M}$ and a group of $q$ agents deserves to be represented by $q \frac{k}n = qr-\binom{q}{2}+1$ matchings. Moreover, note that only node agents can be part of a violating group, as good agents approve all matchings and dummy agents approve more than $\frac{2}{3}$ of the matchings (which is enough to show the claim, as every cohesive group can have size at most $2q\leq \frac{2}{3}n$).  In the following, we show that there exists a size-$q$ clique in $G$ if and only if $\mathcal{M}$ does not satisfy PJR. Intuitively, this holds as for a group of node agents $X=\{a_v\mid v\in V'\}$, the set of matchings approved by some agent from $X$ corresponds to the set of edges that are incident to some node from $V'$. 

    $``\Rightarrow"$ Let $V'$ be a clique in $G$ of size $q$, then exactly $qr-\binom{q}{2}$ different edges are incident to some node from $V'$ (every node is incident to $r$ edges and $\binom{q}{2}$ edges have both endpoints in $V'$). As there exist $q$ good agents, the group $\{a_v\mid v\in V'\}$ is $(qr-\binom{q}{2}+1)$-cohesive. As they together approve only $qr-\binom{q}{2}$ different matchings from $\mathcal{M}$, $\{a_v\mid v\in V'\}$ is a violating group for PJR. 
 
    $``\Leftarrow"$ Assume that there exists a violating group of agents $X$ for PJR. Recall that only node agents can be part of $X$. Moreover, as only the $q$ good agents are approved by some node agent, it further needs to hold that $|X|\leq q$. For the sake of contradiction, assume that $|X|=x$ for some $x<q$. Each set of vertices of size $x$ in $G$ needs to be adjacent to at least $xr-\binom{x}{2}$ different edges. Thus, agents from $X$ must approve at least $xr-\binom{x}{2}$ different matchings in $\mathcal{M}$, while they deserve to be represented by  $x \cdot (r-\frac{q-1}{2}+\frac{1}{q})$ matchings. However, note that such a group cannot be violating, as, for all $x\in[1,q-1]$, it holds that
    \begin{align*}
      &&xr-\binom{x}{2}&>&x\cdot(r-\frac{q-1}{2}+\frac{1}{q}) \\
     \Leftrightarrow&& r-\frac{x-1}{2}&>&r-\frac{q-1}{2}+\frac{1}{q} \\
    \Leftrightarrow &&-x&>&-q+\frac{2}{q},
    \end{align*} where the last inequality holds as $x\in [1,q-1]$ and ${q>3}$. Thus, $X$ needs to have size $q$. For a group of size $q$ to violate PJR, they need to approve at most $qr-\binom{q}{2}$ matchings together. Thus, the set of vertices $\{v\mid a_v\in X\}$ is incident to at most $qr-\binom{q}{2}$ different edges in $G$ implying that they form a clique in $G$.
\end{proof}

We note that the reduction for \Cref{pr:check-core} described below makes use of committees containing matchings that are not approved by any agent and are thus not Pareto optimal. In order to avoid this technical problem, we describe at the end of the proof how the reduction can be altered such that all considered matchings are Pareto optimal. 
\checkcore*

\begin{proof}
In the NP-hard \textsc{Exact Cover by 3-Sets} (\textsc{X3C}) problem, we are given a universe $X$ of size $3q$ and a collection $C$ of $3$-element subsets of $X$ and the question is whether there exists an exact cover $C'\subseteq C$ of $X$. In fact, we reduce from the restricted version where each element appears in exactly three sets from $C$. Thus, it holds that $|C|=3q$. We construct a matching election and a committee $\mathcal{M}$ of size $k=6$ as follows.

We start by describing the central part of the constructed matching election before adding additional agents to cope with some technical details.  For each element $x\in X$, we insert one \emph{element agent} $a_x$ and one \emph{dummy element agent} $b_x$. Moreover, for each set $c\in C$, we add one \emph{set agent} $a_c$. Approvals are symmetric. For each element $x\in X$, the element agent $a_x$ and the dummy element agent $b_x$ approve each other. Moreover, the element agent $a_x$ approves the three set agents $a_c$ corresponding to sets in which it is contained. 
We construct $\mathcal{M}$ such that each dummy element agent approves one matching, each element agent approves two matchings, and each set agent approves two matchings. Moreover, we modify the instance such that each possible blocking coalition needs to deserve to be represented by three matchings and needs to contain $7q$ of the so-far introduced agents.

\tikzstyle{mnode}=[draw, circle, fill, inner sep = 0.75pt]
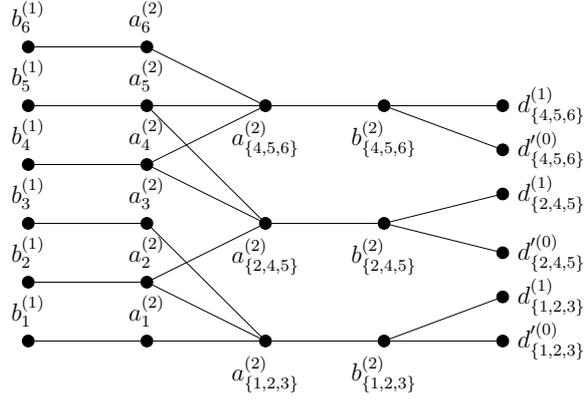
\begin{figure}
    \centering
   \resizebox{.5\textwidth}{!}{%
   \begin{tikzpicture}
        \node[node, label=90:$a_1^{(2)}$] (a1) at (-3, 0) {};
        \node[node, label=90:$a_2^{(2)}$] (a2) at (-3, 1) {};
        \node[node, label=90:$a_3^{(2)}$] (a3) at (-3, 2) {};
        \node[node, label=90:$a_4^{(2)}$] (a4) at (-3, 3) {};
        \node[node, label=90:$a_5^{(2)}$] (a5) at (-3, 4) {};
         \node[node, label=90:$a_6^{(2)}$] (a6) at (-3, 5) {};
         \node[node, label=90:$b_1^{(1)}$] (b1) at (-5, 0) {};
        \node[node, label=90:$b_2^{(1)}$] (b2) at (-5, 1) {};
        \node[node, label=90:$b_3^{(1)}$] (b3) at (-5, 2) {};
        \node[node, label=90:$b_4^{(1)}$] (b4) at (-5, 3) {};
        \node[node, label=90:$b_5^{(1)}$] (b5) at (-5, 4) {};
         \node[node, label=90:$b_6^{(1)}$] (b6) at (-5, 5) {};
         \node[node, label=270:$a_{\{1,2,3\}}^{(2)}$] (v1) at (-1, 0) {};
         \node[node, label=270:$a_{\{2,4,5\}}^{(2)}$] (v2) at (-1, 2) {};
         \node[node, label=270:$a_{\{4,5,6\}}^{(2)}$] (v3) at (-1, 4) {};
         
         \node[node, label=270:$b_{\{1,2,3\}}^{(2)}$] (vv1) at (1, 0) {};
         \node[node, label=270:$b_{\{2,4,5\}}^{(2)}$] (vv2) at (1, 2) {};
         \node[node, label=270:$b_{\{4,5,6\}}^{(2)}$] (vv3) at (1, 4) {};
         
         \node[node, label=0:$d_{\{4,5,6\}}^{(1)}$] (z1) at (3, 4) {};
         \node[node, label=0:$d'^{(0)}_{\{4,5,6\}}$] (zz1) at (3,3.25) {};
         \node[node,label=0:$d'^{(0)}_{\{2,4,5\}}$] (z2) at (3, 1.5) {};
         \node[node, label=0:$d_{\{2,4,5\}}^{(1)}$] (zz2) at (3,2.5) {};
         \node[node, label=0:$d_{\{1,2,3\}}^{(1)}$] (z3) at (3, 0.75) {};
         \node[node, label=0:$d'^{(0)}_{\{1,2,3\}}$] (zz3) at (3,0) {};
         \draw (a1) edge (b1);
         \draw (a2) edge (b2);
         \draw (a3) edge (b3);
         \draw (a4) edge (b4);
         \draw (a5) edge (b5);
         \draw (a6) edge (b6);
         \draw (a1) edge (v1);
         \draw (a2) edge (v1);
         \draw (a2) edge (v2);
         \draw (a3) edge (v1);
         \draw (a4) edge (v2);
         \draw (a4) edge (v3);
         \draw (a5) edge (v2);
         \draw (a5) edge (v3);
         \draw (a6) edge (v3);
          \draw (v1) edge (vv1);
          \draw (v2) edge (vv2);
          \draw (v3) edge (vv3);
        \draw (vv3) edge (z1);
         \draw (vv3) edge (zz1);
         \draw (vv2) edge (z2);
         \draw (vv2) edge (zz2);
         \draw (vv1) edge (z3);
         \draw (vv1) edge (zz3);
    \end{tikzpicture} }
    \caption{Example of the hardness reduction from \Cref{pr:check-core} for \textsc{Exact Cover By 3-Sets} instance: $X=\{1,2,3,4,5,6\}$ and $\{\{1,2,3\},\{2,4,5\},\{4,5,6\}\}\subseteq C$. Numbers in the superscripts denotes the number of matchings from $\mathcal{M}$ the agent approves.}
    \label{fi:check-core}
\end{figure}

To realize these requirements, we need to introduce several additional agents. That is, we introduce for each set $c\in C$, three \emph{dummy set} agents $b_c$, $d_c$, and $d'_c$. Approvals are again symmetric. Agent $b_c$ approves the set agent $a_c$ and the two dummy agents $d_c$ and $d'_c$. We construct $\mathcal{M}$ such that $b_c$ approves two matchings, $d_c$ approves one matching and $d'_c$ approves zero matchings.  Lastly, to adjust the total number of agents, we add $14q$ \emph{filling} agents with empty approval ballot. In total, the instance consists of $6q$ element and dummy element agents, $3q$ set agents and $9q$ dummy set agents and $14q$ filling agents, i.e., $32q$ agents in total. For a visualization of the reduction see \Cref{fi:check-core}.

We are now ready to construct $\mathcal{M}$ realizing the already mentioned happiness scores of the agents. First, we add a matching where for each element $x\in X$, the element agent $a_x$ is matched to the dummy element agent $b_x$ and, for each set $c\in C$, the set agent $a_c$ is matched to the dummy set agent $b_c$. In the second matching, we match all element agents $a_x$ to a set agent $a_c$ that they approve. (Note that such a perfect matching of element agents and set agents has to exist because these agents form a $3$-regular bipartite graph.) Moreover, for each $c\in C$, we match dummy set agents $b_c$ and $d_c$. Finally, we add four matchings that are not approved by anyone. Thus, as $\mathcal{M}$ consists of six matchings and as the total number of agents is $32q$, each group of $\frac{32q}{6} = \frac{16q}{3}$ agents deserves to be represented by one matching. We now show that the given \textsc{X3C} instance $(X,C)$ admits an exact cover if and only if there exists a group violating core stability in the constructed matching election.

$``\Rightarrow"$ Let us assume that there exists an exact cover $C'\subseteq C$ of $X$. We claim that the group $S$ consisting of all element and dummy element agents, all set agents corresponding to sets from $C'$, and all dummy set agents block committee $\mathcal{M}$. Note that $S$ consists of $6q+1q+9q=16q$ agents and thus deserves to be represented by three matchings. We now describe the three blocking matchings. For each $c\in C$, $b_c$ is matched to $d_c$ in the first two of the three matchings and to $d'_c$ in the third.  For each $c=\{x_i,x_j,x_k\}\in C'$, we match $a_c$ to $a_{x_i}$ in the first matching, to $a_{x_j}$ in the second matching, and to $a_{x_k}$ in the third matching. This is always possible, as $C'$ is an exact cover of $X$. Thereby, each element agent is matched to a set agent in one of the three matchings. We match each element agent in the remaining two matchings to the corresponding dummy element agent. Note that all element agents and all set agents corresponding to sets from $C'$ approve all three matchings. All dummy element agents approve two matchings. For all $c\in C$, $b_c$ approves all three matchings, $d_c$ approves two matchings and $d'_c$ one matching. Thus, $S$ is blocking. 

$``\Leftarrow"$ Assume that there exists a blocking coalition $S$ for $\mathcal{M}$ because of a multiset of matchings $\mathcal{M}'$. Note that there exist only $3q$ non-filling agents that do not approve any matching from $\mathcal{M}$ and only $9q$ non-filling agents that approve at most one matching from $\mathcal{M}$. As each group of $\frac{16q}{3}$ agents deserves to be represented by one matching and $3q\cdot\frac{3}{16q}<1$ and $9q\cdot\frac{3}{16q}<2$, it needs to hold that $|\mathcal{M}'|\geq 3$ and thus $S$ needs to have size at least $16q$. Moreover, note that there cannot exist a blocking coalition that deserves to be represented by four matchings, as there exist only $18q$ non-filling agents. 

To complete the proof, we need the following claim.
\begin{claim*}
Let $S_{\text{set}}\subseteq S$ be the set of set agents $a_c$ that are part of the blocking coalition $S$. Then, it holds that $|S_{\text{set}}|= q$. 
\end{claim*}
\begin{claimproof} 
As $S$ can only contain non-filling agents, from $|S|\geq 16q$ and the fact that there exist only $18q$ non-filling agents of which $3q$ are set agents, it follows that $|S_{\text{set}}|\geq q$ needs to hold. To prove that $|S_{\text{set}}|= q$, first of all, note that all set agents from $S_{\text{set}}$ need to approve all three matchings from $\mathcal{M}'$. Thus, in total, there exist $3|S_{\text{set}}|$ pairs in $\mathcal{M}'$ each containing exactly one agent from $S_{\text{set}}$. 
Let $w$ be the number of dummy element agents that are part of $S$, $y$ the number of dummy set agents of the form $d_c$ and $z$ the number of dummy set agents of the form $d'_c$. Overall, it needs to hold that $|S_{\text{set}}|+w+y+z+6q\geq 16q$ and thus $|S_{\text{set}}|+w+y+z\geq 10q$. Note that as each dummy element agent only needs to approve two matchings from $\mathcal{M}'$, even if $w=3q$, each element agent can be matched to an agent from $S_{\text{set}}$ in one matching. For the sake of contradiction, let us assume that $t:=S_{\text{set}}-q>0$. Then, $3t$ approvals for set agents that do not come from element agents which are matched to the corresponding dummy element agent in the other two matchings are needed. However, for each two of these $3t$ approvals, either an element agent needs to be matched more than once to a set agent or an dummy set agent $b_c$ needs to be matched twice to a set agent. While the former implies that the corresponding dummy element agent cannot be part of the blocking coalition $S$, the latter implies that either one less dummy set agent of the form $d_c$ or $d'_c$ can be part of the blocking coalition. Thus, $t>0$ implies that $w+y+z\leq 9q-\frac{3}{2}t$. Overall we get that  $q+t+w+y+z\leq q+t+9q-\frac{3}{2}t=10q-\frac{1}{3}t$. Thus, it needs to hold that $t=0$. This directly implies that $|S_{\text{set}}|=q$.
\end{claimproof}
From the claim it directly follows that $S$ consists of the agents $S_{\text{set}}$ and all non-filling agents that are not set agents. 

To ensure that all dummy element agents approve two matchings from $\mathcal{M}'$, each element agent needs to be matched to the corresponding dummy element agent in two of the three matchings. Moreover, each set agent from $S_{\text{set}}$ needs to approve all three matchings from $\mathcal{M}'$ and no dummy set agent $b_c$ can be matched to an agent from $S_{\text{set}}$. Thus, each element agent is matched to a set agent it approves in exactly one of the three matchings. As each set agent from $S_{\text{set}}$ needs to approve all three matchings from $\mathcal{M}'$, this implies that each set agent from $S_{\text{set}}$ needs to be matched to each of the three element agents corresponding to its elements in one of the three matchings. Thus, $S_{\text{set}}$ forms an exact cover of $X$.

\medskip
It is possible to slightly modify the reduction to avoid that Pareto-dominated matchings are part of the given committee, at the cost of losing symmetry. We start by modifying the approval ballots of $6q$ arbitrary filling agents and make them approve all element agents $a_x$ for $x\in X$ and all dummy set agents $b_c$ for $c\in C$ (but not the other way round). Constructing~$\mathcal{M}$, instead of adding four matchings not approved by anyone, we add four matchings in which the $6q$ modified filling agents are matched to all element agents $a_x$ and dummy set agents $b_c$. Note that these matchings are Pareto optimal, as modified filling agents only approve these agents and the remaining non-filling agents also only approve element agents $a_x$ or dummy set agents $b_c$.

The correctness of the forward direction of the proof remains unaffected, while for the backward direction it is necessary to argue why none of the modified filling agents can be part of a blocking coalition. To see this, note that these agents approve four matchings in $\mathcal{M}$ and thus any blocking coalition $S$ they are part of needs to deserve to be represented by at least five matchings. However, this implies that $|S|\geq  5\cdot \frac{16q}{3} > 26q$, which cannot be the case, as there exist only $24q$ agents approving some other agent.
\end{proof}
\end{document}